\documentclass{article}
\usepackage{natbib}
\usepackage{fullpage}

\usepackage{microtype}
\usepackage{graphicx}
\usepackage{subfigure}
\usepackage{booktabs}

\usepackage{hyperref}
\hypersetup{
	colorlinks=true,
	linkcolor=red,
	citecolor=blue,
	urlcolor=violet,
	pdftitle={Optimal bounds for \texorpdfstring{$\ell_p$}{Lp} sensitivity sampling via \texorpdfstring{$\ell_2$}{L2} augmentation},
	linktocpage=true,
}

\usepackage{amsmath}
\usepackage{amssymb}
\usepackage{mathtools}
\usepackage{amsthm}

\usepackage[capitalize,noabbrev]{cleveref}

\theoremstyle{plain}
\newtheorem{theorem}{Theorem}[section]
\newtheorem{proposition}[theorem]{Proposition}
\newtheorem{lemma}[theorem]{Lemma}
\newtheorem{corollary}[theorem]{Corollary}
\theoremstyle{definition}
\newtheorem{definition}[theorem]{Definition}
\newtheorem{question}[theorem]{Question}

\theoremstyle{remark}

\DeclarePairedDelimiter\parens{\lparen}{\rparen}
\DeclarePairedDelimiter\abs{\lvert}{\rvert}
\DeclarePairedDelimiter\norm{\lVert}{\rVert}

\DeclarePairedDelimiter\braces{\lbrace}{\rbrace}
\DeclarePairedDelimiter\bracks{\lbrack}{\rbrack}

\newcommand\range{\mathrm{range}}

\newcommand\ranges{\mathrm{ranges}}
\newcommand\rng[1]{\mathrm{range}_{#1}}

\DeclareMathOperator{\poly}{poly}

\let\Pr\relax\DeclareMathOperator*{\Pr}{\mathbf{Pr}}

\DeclareMathOperator*{\E}{\mathbf{E}}

\DeclareMathOperator{\diam}{diam}
\DeclareMathOperator{\erfc}{erfc}

\newcommand{\eps}{\varepsilon}
\renewcommand{\epsilon}{\varepsilon}

\usepackage[textsize=tiny]{todonotes}

\begin{document}

\title{Optimal bounds for \texorpdfstring{$\ell_p$}{Lp} sensitivity sampling via \texorpdfstring{$\ell_2$}{L2} augmentation}

\author{Alexander Munteanu\thanks{Dortmund Data Science Center, Faculties of Statistics and Computer Science, TU Dortmund University, Dortmund, Germany. Email: \texttt{alexander.munteanu@tu-dortmund.de}.}
\and Simon Omlor \thanks{Faculty of Statistics and Lamarr-Institute for Machine Learning and Artificial Intelligence, TU Dortmund University, Dortmund, Germany. Email: \texttt{simon.omlor@tu-dortmund.de}.}}

\vskip 0.3in

\maketitle
\allowdisplaybreaks
\begin{abstract}
Data subsampling is one of the most natural methods to approximate a massively large data set by a small representative proxy. In particular, sensitivity sampling received a lot of attention, which samples points proportional to an individual importance measure called sensitivity. This framework reduces in very general settings the size of data to roughly the VC dimension $d$ times the total sensitivity $\mathfrak S$ while providing strong $(1\pm\varepsilon)$ guarantees on the quality of approximation. The recent work of \citet{woodruffyasuda23} improved substantially over the general $\tilde O(\varepsilon^{-2}\mathfrak Sd)$ bound for the important problem of $\ell_p$ subspace embeddings to $\tilde O(\varepsilon^{-2}\mathfrak S^{2/p})$ for $p\in[1,2]$. Their result was subsumed by an earlier $\tilde O(\varepsilon^{-2}\mathfrak Sd^{1-p/2})$ bound which was implicitly given in the work of \citet{ChenD21}. 
We show that their result is tight when sampling according to plain $\ell_p$ sensitivities. We observe that by augmenting the $\ell_p$ sensitivities by $\ell_2$ sensitivities, we obtain better bounds improving over the aforementioned results to optimal linear $\tilde O(\varepsilon^{-2}(\mathfrak S+d)) = \tilde O(\varepsilon^{-2}d)$ sampling complexity for all $p \in [1,2]$. In particular, this resolves an open question of \citet{woodruffyasuda23} in the affirmative for $p \in [1,2]$ and brings sensitivity subsampling into the regime that was previously only known to be possible using Lewis weights \citep{CohenP15}. 
As an application of our main result, we also obtain an $\tilde O(\varepsilon^{-2}\mu d)$ sensitivity sampling bound for logistic regression, where $\mu$ is a natural complexity measure for this problem. This improves over the previous $\tilde O(\varepsilon^{-2}\mu^2 d)$ bound of \citet{MaiMR21} which was based on Lewis weights subsampling.
\end{abstract}

\clearpage
\tableofcontents
\clearpage

\section{Introduction}
Massive data sets have become ubiquitous in recent years and standard machine learning approaches reach the limits of tractability when these large data need to be analyzed. Given a data matrix $A\in\mathbb{R}^{n\times d}$ where the number of data points exceeds the number of features by a large margin, i.e., $n\gg d$, a popular approach to address the computational limitations is to subsample the data \citep{Munteanu23}. While uniform sampling is widely used in practice, it can be associated with large loss of approximation accuracy for various machine learning models. Thus, a lot of work has been dedicated to the design of importance sampling schemes, to subsample points proportional to some sort of importance measure for the contribution of individual points, such that important points become more likely to be in the sample, while less important or redundant points are sampled with lower probability \citep{MunteanuS18}. 

\textbf{General setting.}
The arguably most general and popular importance sampling approach is called the sensitivity framework \citep{FeldmanL11,FeldmanSS20}. It aims at approximating functions applied to data that can be expressed in the following way: consider an individual non-negative loss function $h_i\colon \Omega \rightarrow \mathbb{R}_{\geq 0}$ for each row vector $a_i, i\in[n]:= \{1,2,\ldots,n\}$ of the data matrix $A$, more precisely we set $h_i(x)=h(a_i x)$ for some function $h\colon \mathbb{R} \rightarrow \mathbb{R}_{\geq 0} $. The problem we consider is approximating
\begin{align}\label{loss}
    f(x)\coloneqq \sum_{i\in[n]} h_i(x)
\end{align}
over all $x$ in the domain $\Omega$, by subsampling and reweighting the individual contributions. Formally, we obtain $S\subset [n]$ with $|S|\ll n$, and $w_i> 0$ for all $i\in S$ and define
\begin{align}\label{surrogate_loss}
    \tilde f(x)\coloneqq \sum_{i\in S} w_i h_i(x)
\end{align}
to be the surrogate loss, that given an approximation parameter $\eps\in (0,1)$ should be a pointwise $(1\pm\eps)$-approximation for the original loss function. More specifically, we require that
\begin{align}\label{approx_guarantee}
   \forall x\in \Omega\colon \tilde f(x)=(1\pm\eps) f(x).
\end{align}
The surrogate function \eqref{surrogate_loss} can then be used in downstream machine learning tasks, such as optimization, with very little bounded errors, and can be processed much more efficiently than the original loss. Note, that most empirical risk minimization problems or negative log-likelihood functions can be expressed as in \cref{loss}.

\textbf{Importance subsampling.} The sensitivity framework provides a way of obtaining the guarantee of \cref{approx_guarantee} by first computing individual sensitivities for $i\in [n]$
\begin{align}\label{sensitivity}
   \varsigma_i = \sup_{x\in\Omega} \frac{h_i(x)}{\sum_{j\in[n]} h_j(x)},
\end{align}
where we let $\varsigma_i = 0$ when the denominator is zero. We note that computing the sensitivity is usually intractable, but it suffices to obtain close approximations, which can be done efficiently for many important problems.

Then, subsampling $S\subset [n]$ with probabilities $p_i$ proportional to $\varsigma_i$ and reweighting the individual contributions by $w_i=1/(p_i|S|)$, preserves the objective function in expectation and when $|S|=\tilde O(\eps^{-2}\mathfrak Sd)$, we obtain by a concentration result that the $(1\pm\eps)$-approximation defined in \cref{approx_guarantee} holds with constant probability. Here, $d$ is the VC dimension of a set system associated with the functions $h_i$, and $\mathfrak S=\sum_{i\in[n]} \varsigma_i$ denotes their total sensitivity.

\textbf{Subspace embeddings for $\ell_p$.}
In this paper, we focus on applying the above framework to the more specific, yet very versatile problem of constructing a so-called $\ell_p$ subspace embedding via subsampling. Given a data matrix $A\in\mathbb{R}^{n\times d}$ with \emph{row-vectors} $a_i,i\in[n]$, and a norm parameter $p \in[1,\infty]$, our goal is to calculate probabilities $p_i$, and weights $w_i$ for each $i\in[n]$, such that a small random subsample $S\subset [n]$ according to this distribution satisfies
\begin{align}\label{lp_approx_guarantee}
   \forall x\in \mathbb{R}^d\colon \sum_{i\in S} w_i |a_i x|^p=(1\pm\eps) \sum_{i\in [n]}|a_i x|^p
\end{align}
with probability at least $1-\delta$.

We remark at this point that since the $\ell_p$ norms are absolutely homogeneous, previous work usually includes the weights by folding them into the data, i.e., $w_i|a_ix|^p$ becomes $|(w^{1/p}a_i)x|^p$. This usually allows to simplify large parts of the technical analysis by reducing to the unweighted case. We explicitly \emph{do not} use this trick for reasons to be discussed later.

With the definition of our main problem in place, we next define the sensitivity measure associated with this loss function and note that they are also called $\ell_p$ leverage scores in our scope.
\begin{definition}[$\ell_p$-sensitivities/-leverage scores]
    Let $A \in \mathbb{R}^{n \times d}$, and $1 \leq p < \infty$. We define the $i$-th $\ell_p$-sensitivity $\varsigma_i^{(p)}$ or $\ell_p$-leverage score $l_i^{(p)}$ of $A$ to be
    \[
        \varsigma_i^{(p)}=l_i^{(p)}=\sup\nolimits_{x \in \mathbb{R}^d\setminus\{0\}} \frac{|a_i x|^p}{\norm{Ax}_p^p},
    \] and the total $\ell_p$ sensitivity as $\mathfrak S^{(p)} \colon= \sum_{i\in[n]}\varsigma_i^{(p)}$.
\end{definition}
It can be shown in some settings, that the total sensitivity is a lower bound on the required number of samples \citep[cf.][]{TolochinskyJF22}, and $d$ is a natural lower bound for our problem because this number of samples is required to even preserve the rank of $A$ under subsampling. It is also known that $\mathfrak S^{(p)}\leq d$ for all $p\in[1,2]$ \citep[cf.][]{woodruffyasuda23}. Thus, a natural question to ask is whether it is possible to obtain the guarantee of \cref{lp_approx_guarantee} within roughly $\mathfrak S^{(p)}+d = O(d)$ samples.

A plain application of the sensitivity framework requires $\tilde O(\eps^{-2}\mathfrak Sd)$ samples, and a lower bound of $\tilde\Omega(d/\eps^2)$ was given by \citet{LiWW21} in a broader setting against any data structure that answers $\ell_p$ subspace queries. The sensitivity sampling upper bound thus seems off by a factor $d$, and the recent work of \citet{woodruffyasuda23} has made significant progress by showing that one can do better for all $p\in(1,4)$. They obtained near optimal bounds when $p$ becomes close to $p=2$. In this particular case, a $\tilde O(\eps^{-2}\mathfrak S) = \tilde O(\eps^{-2}d)$ bound was known before to be possible by $\ell_2$ leverage score sampling \citep{Mahoney11}. However, near the boundary $p=1$ (and for $p>4$) the bounds of \citet{woodruffyasuda23} are again off by a factor $d$. In the case $p<2$, the previous work of \citet{ChenD21} gave better results that were achieved implicitly by relating $\ell_p$ leverage scores to Lewis weights. Again, these bounds become worse when $1 \leq p\ll 2$ and are off by a $\sqrt{d}$ factor for $p=1$. See below for a more detailed comparison between these works, and our work.

This brings us to the main open question in this line of research, which is also the central question studied in our paper:

\begin{question}\label{main_question}
    \textit{Is it possible to achieve a sampling complexity of $\tilde O(\eps^{-2}(\mathfrak S + d))$ for constructing an $\ell_p$ subspace embedding, see \cref{lp_approx_guarantee}, via sensitivity sampling?}
\end{question}

\subsection{Our contribution}
Our first contribution is an $\tilde\Omega(d^{2-p/2})$ lower bound against \emph{pure} $\ell_p$ leverage score sampling, i.e., when the sampling probabilities are proportional to $l_i^{(p)}$.
\begin{theorem}[Informal restatement of \cref{thm: lowerbound}]\label{thm:lowerbound_informal}
    There exists a matrix $A\in\mathbb{R}^{m\times 2d}$, for sufficiently large $m\gg 2d$, such that if we sample each row $i \in [n]$ with probability $p_i:= \min \{1, k l_i^{(p)}\}$ for some $k \in \mathbb{N}$, then with high probability, the $\ell_p$ subspace embedding guarantee (see \cref{lp_approx_guarantee}) does not hold unless $k=\Omega(d^{2-p/2}/(\log d)^{p/2})$.
\end{theorem}
The proof in the appendix is conducted by constructing a matrix with two parts. Subsampling at least $d$ rows of each part is necessary to even preserve the rank. However, the total $\ell_p$ sensitivity of one part is significantly larger than the total sensitivity of the other part, roughly by a factor of $d^{1-p/2}$. This implies that $\ell_p$ leverage score sampling requires oversampling by that factor to collect the required rows from both parts, which yields the lower bound.

In particular, \cref{thm:lowerbound_informal} proves that the upper bounds implied by \citet{ChenD21} for pure $\ell_p$ leverage score sampling are optimal in the worst case, settling the complexity of pure $\ell_p$ sampling up to polylogarithmic factors.
It also proves that we necessarily need to change the sampling probabilities to get below the lower bound towards optimal linear $d$ dependence. We note that augmentation by uniform sampling does not help, since the size of the high sensitivity part can be increased to give the same imbalance between the number of rows of the two parts as between their sampling probabilities in pure $\ell_p$ sampling.

Our main result is the following sensitivity sampling bound that is achieved by augmenting $\ell_p$ with $\ell_2$ sensitivities, a new technique that we call ``$\ell_2$ augmentation''. The theorem answers \cref{main_question} in the affirmative for all $p \in[1,2]$, improving over \citet{ChenD21,woodruffyasuda23}, and settling the complexity of $\ell_p$ subspace embeddings via sensitivity sampling in this regime. We leave the case $p>2$ as an open question, where our work resembles the same bounds as \citet{woodruffyasuda23}.

Our theorem is generalized to also hold for the $\ell_p$ variant of the ReLU function where $h(r)=\max\{0,t\}^p$. For this loss function, we also require a data dependent mildness parameter $\mu=\mu(A)$ to be bounded. 
\begin{definition}[$\mu$-complexity, \citealp{MunteanuSSW18,MunteanuOP22}, slightly modified]\label{def:mu-complex_informal}
    Given a data set $A \in \mathbb{R}^{n \times d}$ we define the parameter $\mu\coloneqq\mu(A)$ as 
    \[\mu(A)=\sup\nolimits_{\norm{Ax}_p=1}\frac{\norm{Ax}_p^p}{\norm{(Ax)^-}_p^p} = \sup\nolimits_{\norm{Ax}_p=1}\frac{\norm{Ax}_p^p}{\norm{(Ax)^+}_p^p},
    \]
    where $(Ax)^+$, and $(Ax)^-$ are the vectors comprising only the positive resp. negative entries of $Ax$ and all others set to $0$.\footnote{To see that the suprema are equal, one can flip signs: for any $x$ the suprema range over, we have that $\frac{\Vert Ax \Vert_p^p}{\Vert (Ax)^+ \Vert_p^p} = \frac{\Vert A(-x) \Vert_p^p}{\Vert (A(-x))^- \Vert_p^p}$.}
\end{definition}
For the main problem where $h(r)=|t|^p$, we can remove this parameter or equivalently set $\mu=1$.

\begin{theorem}[Informal restatement of \cref{thm:samplingthm}]\label{thm:samplingthm_informal}
    Let $A\in\mathbb{R}^{n\times d}$. Let $\varepsilon, \delta \in (0, 3/10)$, $p \in [1, 2]$. Consider $f(Ax)=\sum_{i=1}^n h(a_i x)$ for $h(r)=|r|^p$, where we set $\mu=1$, or $h(r)=\max\{0,r\}^p$, in which case $\mu=\mu(A)$, see \cref{def:mu-complex_informal}.
    Set $\alpha= O\left(\frac{ (\log(d \mu\log(1/\delta)/\varepsilon)\log^2d+\log(1/\delta)}{\varepsilon^2}\right).$
    Let $S\subset [n]$ be a sample of size $m = O(d \mu \alpha)$
    \begin{align*}
        = O\left(\frac{d\mu}{\varepsilon^2}(\log(d \mu\log(\delta^{-1})/\varepsilon)\log^2d+\log(\delta^{-1}))\right),
    \end{align*}
    where the probability for any $i \in [n]$ satisfies $$p_i:= \Pr(i\in S) \geq \min\left\{ 1, \alpha\left(\mu l_i^{(p)}+l_i^{(2)}+ \frac{1}{n}\right)\right\},$$ 
    and the corresponding weight is set to $w_i=1/p_i$.
    Then with probability at least $1-\delta$ it holds for all $x\in \mathbb{R}^d$ simultaneously that 
    \begin{align*}
        \left| \sum_{i\in S} w_i h(a_i x) - \sum_{i\in [n]} h(a_i x)\right| \leq \eps \sum_{i\in [n]}h(a_i x).
    \end{align*}
\end{theorem}

As an application, we extend our main result to hold for the logistic loss function $h(t)=\ln(1+\exp(t))$ as well. This improves over the previous $\tilde O(\varepsilon^{-2}\mu^2 d)$ bound of \citet{MaiMR21} which was based on Lewis weights subsampling.

\begin{theorem}[Informal restatement of \cref{thm:logistic}]\label{thm:logistic_informal}
     Let $\varepsilon, \delta \in (0, 3/10)$. Let $A\in \mathbb{R}^{n\times d}$ with $\mu=\mu(A)$
     , see \cref{def:mu-complex_informal}. Consider $f(Ax)=\sum_{i=1}^n h(a_i x)$ with the logistic loss $h(t)=\ln(1+\exp(t))$.
     Further assume that we sample with probabilities $p_i \geq \min\{1,\alpha (\mu l_i^{(1)}+l_i^{(2)}+\frac{\mu d}{n})\}$ for $\alpha \geq O( ( \log^3(d\mu\log(\delta^{-1})/\varepsilon ) + \ln(\delta^{-1}) ) /\varepsilon^{2})$, where the number of samples is 
     \[
        m= O\left(\frac{d\mu}{\varepsilon^{2}} \left( \log^3(d \mu \log(\delta^{-1})/\varepsilon) + \log(\delta^{-1}) \right) \right).
     \]
     Then with probability at least $1-\delta$ it holds for all $x\in \mathbb{R}^d$ simultaneously that 
     \[
        \left| \sum_{i\in S} w_i h(a_i x) - \sum_{i\in [n]} h(a_i x)\right| \leq \eps \sum_{i\in [n]}h(a_i x).
     \]
\end{theorem}

\subsection{Comparison to related work}
For $\ell_p$ subspace embeddings, the total sensitivity is bounded by $d$, for $p\in [1,2]$, and by $d^{p/2}$ for $p\in (2,\infty)$, \citep[cf.][]{woodruffyasuda23}. It is known that using so-called $\ell_p$ Lewis weights, we can subsample a nearly optimal amount $\tilde O(\eps^{-2}d)$ respectively $\tilde O(\eps^{-2}d^{p/2})$ of rows to obtain the $\ell_p$ subspace embedding guarantee of \cref{lp_approx_guarantee}, \citep[see][]{CohenP15,WoodruffY23lewis}. Recent work of \citet{JambulapatiLLS23} recovers matching bounds via a novel sampling distribution, and for a broad array of semi-norms. 
On the other hand, using $\ell_p$ sensitivity sampling in a plain application of the sensitivity framework requires $\tilde O(\eps^{-2}\mathfrak S^{(p)}d)$, which is off by a factor $d$ in the worst case for any value of $p\in[1,\infty)$.

Recently improved sensitivity sampling bounds of \citet{woodruffyasuda23} are $\tilde O(\eps^{-2}\mathfrak S^{2/p})$ for $p\in[1,2]$ and $\tilde O(\eps^{-2}\mathfrak S^{2-2/p})$ for $p \in [2,\infty)$. These bounds are much better than the standard bounds when $p\in (1,4)$ is close to the case $p=2$, but they deteriorate towards $p=1$ and $p=4$, where the gap is again a factor of $d$. For $p>4$ their worst case bounds are even worse than the plain sensitivity framework. We note that \citet{woodruffyasuda23} gave an improved bound for this regime as well, albeit not with a direct sensitivity sampling approach. Instead, they gave an algorithm that recursively 'flattens and samples' heavy rows with respect to their sensitivities for $\ell_2$ and $\ell_p$.

Although improving over the standard bounds for $p<2$ as well, the main improvement of \citet{woodruffyasuda23} lies in the case $p>2$. This is because prior to their results, \citet{ChenD21} implicitly showed a bound of $\mathfrak Sd^{1-p/2}$ by relating $\ell_p$ sensitivities to Lewis weights up to an additional factor of $d^{1-p/2}$. Oversampling the sensitivity scores by this amount guarantees that the increased scores exceed the Lewis weights, which in turn implies a sampling complexity of $\tilde O(\eps^{-2}\mathfrak Sd^{1-p/2})$ to be sufficient. We show in \cref{thm:lowerbound_informal} that their bound is tight up to polylogarithmic factors in the worst case, when $\mathfrak S=\Theta(d)$, if we sample according to pure $\ell_p$ leverage scores.

This leads us to our main result, \cref{thm:samplingthm_informal}, which samples for any value of $p\in[1,2]$ a number of $\tilde O(\eps^{-2}(\mathfrak S^{(p)} + \mathfrak S^{(2)})) = \tilde O(\eps^{-2}(\mathfrak S^{(p)} + d))$ many samples according to a mixture of $\ell_p$, and $\ell_2$ sensitivities with a uniform distribution $1/n$. This technique relies only on $\ell_p$ sensitivity sampling and in the worst case, the number of samples amounts to $\tilde O(\eps^{-2}d)$ for all $p\in[1,2]$. We note that this matches up to polylogarithmic factors the optimal complexities obtained by Lewis weights, and by other novel sampling probabilities \citep{CohenP15,WoodruffY23lewis,JambulapatiLLS23}.
See \cref{fig:samplecomplexity} for a visual comparison of our bounds with previous sensitivity sampling bounds.

\begin{figure}[!ht]
\vskip 0.04in
\begin{center}
Sample complexity bounds for $\ell_p$ sensitivity sampling\\
\includegraphics[width=.4\textwidth]{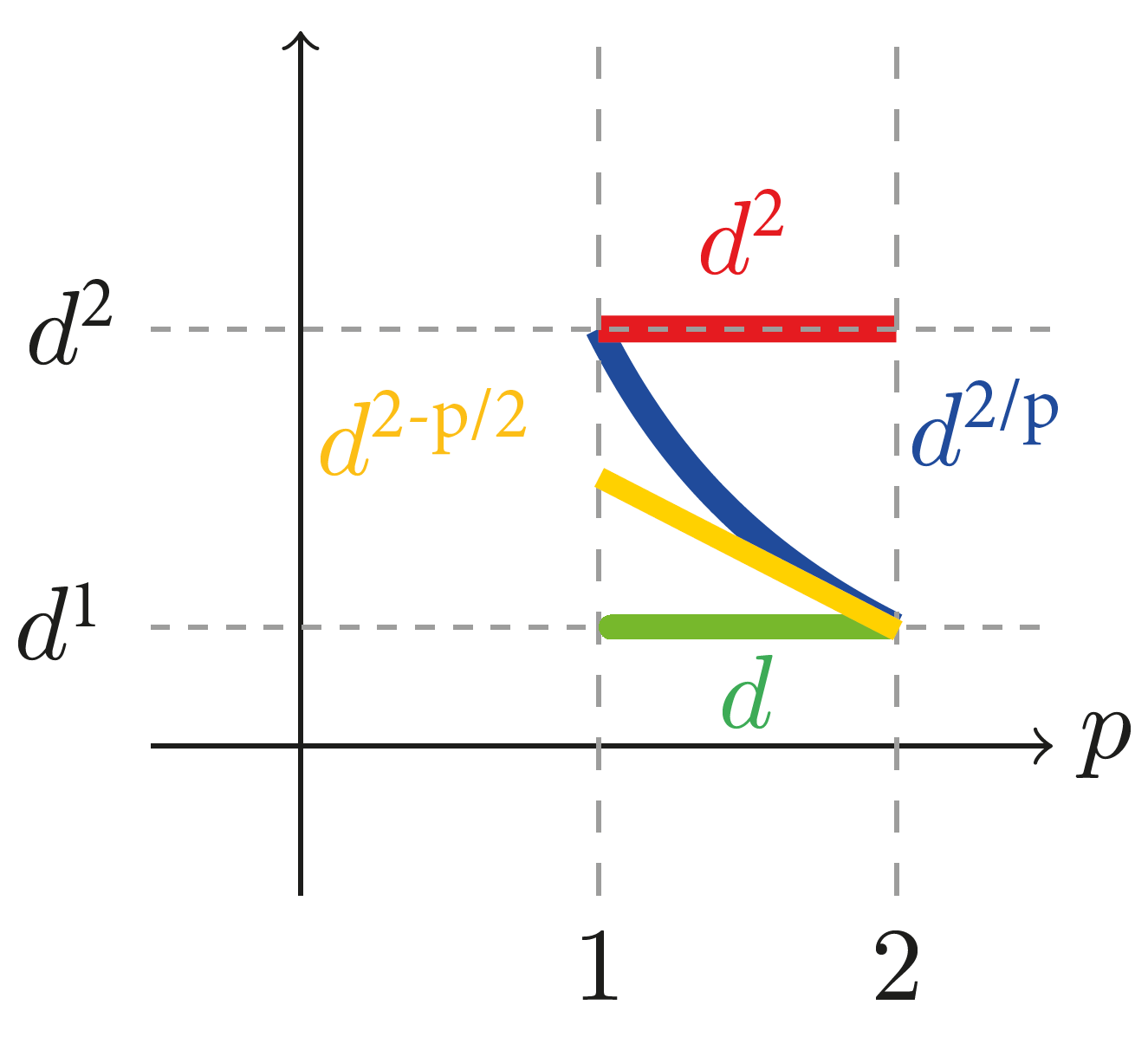}
\vskip -0.1in
\caption{Leading dependence on $d$ for $\ell_p$ sensitivity sampling for $p\in[1,2]$ in the worst case, i.e., when $\mathfrak S^{(p)}=d$. The horizontal axis represents $p$. The vertical axis indicates the exponent on $d$ in the respective sample complexity results. The red line indicates the standard bounds obtained from a plain application of the sensitivity framework \citep{FeldmanSS20}, blue indicates the result of \citet{woodruffyasuda23}, yellow indicates the result of \citet{ChenD21}, and green indicates our new main result.}
\label{fig:samplecomplexity}
\end{center}
\vskip -0.2in
\end{figure}

In particular, note that our bounds improve the $d^2$ respectively $d^{3/2}$ dimension dependence of previous bounds for $p=1$ to linear. This allows an application to logistic regression (\cref{thm:logistic_informal}), obtaining up to polylogarithmic factors a sampling bound of $\tilde O(\eps^{-2}d \mu )$. Previously, the linear dependence on $d$ was only known to be possible using $\ell_1$ Lewis weights \citep{MaiMR21,WoodruffY23lewis}. In fact our result even improves over their $\tilde O(\eps^{-2}d \mu^2 )$ bound by a factor of $\mu$, a complexity parameter introduced by \citet{MunteanuSSW18,MunteanuOP22} for compressing data in the scope of logistic regression and other asymmetric functions.
We note that linear $\mu$ and (near-)linear $d$ dependence was recently achieved in the sketching regime, though at the cost of constant approximation factors \cite{MunteanuOW23}.
We remark that our polylogarithmic dependencies hidden in the $\tilde O$ notation are only $\operatorname{polylog}(\mu,d,\eps^{-1},\delta^{-1})$ and do not depend on $n$, which is also a minor improvement compared to almost all mentioned previous works. 

Our paper assumes that we have access to $\ell_p$ sensitivity scores, without giving details on how to compute or approximate them. We refer to \citep{DasguptaDHKM09,WoodruffZ13,ClarksonDMMMW16,MunteanuOP22} for classic techniques such as ellipsoidal rounding and well-conditioned bases, as well as to recent advances in constructing improved well-conditioned bases \citep{WangW22}, novel $\ell_p$ spanning sets \citep{WoodruffY23spanningset,BhaskaraMV23}, or direct $\ell_p$ sensitivity approximations \citep{PadmanabhanWZ23}. We also refer to \citep{Mahabadi2020,MunteanuO24turnstile} for $\ell_p$ sensitivity sampling in data streams.

One might argue that the sensitivity sampling approach is not very interesting for $p\in[1,2]$, since Lewis weights, among others, already obtain optimal bounds in this regime. However, leverage scores are usually simpler to calculate or to approximate. For instance \citet{CohenP15,MaiMR21} calculate an approximation to Lewis weights by recursively reweighting the data and computing $\ell_2$ leverage scores $O(\log\log(n))$ times over and over again. While the factor $O(\log\log(n))$ overhead seems minor from a theoretical perspective, this slows down computations by a non-negligible amount. We refer to \citep{MaiMR21,MunteanuOP22} for details, where computational issues have been discussed and demonstrated in experiments. 

The experiments in \citep{MaiMR21,MunteanuOP22} also suggest that sensitivity sampling works much better than indicated by upper bounds, sometimes even better than Lewis weights. It is thus very important to find a theoretical explanation for the success of sensitivity sampling and to find out whether they also achieve the optimal complexity or if there are lower bounds preventing them from achieving optimality. This is the motivation behind our work.

We would like to mention that very similar findings have been observed in the center-based clustering regime, where group sampling was known to produce subsamples of optimal size \citep{Cohen-AddadSS21,Cohen-AddadLSSS22,HuangLW22}. But group sampling was often outperformed by the conceptually and computationally simpler sensitivity sampling approach on practical and on hard instances \citep{SchwiegelshohnS22}. Independently of our work, this was explained in \cite{bansal2024sensitivity} by proving that sensitivity sampling also achieves optimal subsample size for $k$-means and $k$-median clustering.

Further, $\ell_p$ sensitivity sampling has already been studied for a plethora of problems such as $\ell_p$ regression \citep{DasguptaDHKM09}, $M$-estimators \citep{ClarksonW15,ClarksonW15focs}, near-convex functions \citep{TukanMF20}, logistic regression \citep{MunteanuSSW18}, other generalized linear models such as probit regression \citep{MunteanuOP22}, Poisson and graphical models \citep{MolinaMK18}, IRT models in educational sciences and psychometrics \citep{FrickKM24}. Also, some seemingly more distant works have strong connections to sensitivity sampling, such as graph sparsification using effective resistances \citep{SpielmanS11}. Our new optimal bounds for all $p\in[1,2]$ cover the most common regime encountered in all of these works and will thus be useful in obtaining improved bounds for a broad array of applications as well.

\section{New sensitivity subsampling bounds}
Our analysis uses several results of \citet{woodruffyasuda23}, and our main argumentation follows their general outline. Since our analysis is in wide parts a strict generalization, the worst possible outcome of our investigations would simply resemble their exact same bounds. As we have indicated previously, this is actually the case for $p>2$ in which our techniques do not improve over their bounds. The corresponding part of the analysis is therefore not contained in our paper and we focus on the $p\in[1,2]$ regime in the remainder.

One main technical argument in this regime is that by monotonicity of maximal $\ell_p$ sensitivities, the largest $\ell_2$ leverage score upper bounds the largest $\ell_p$ leverage score. To leverage this fact, previous analyses relied on an auxiliary subspace embedding for obtaining a constant factor subspace embedding that required $\operatorname{poly}(d)$ overhead.

In our analysis, we bypass this problem by adding the $\ell_2$ leverage scores to the sensitivity upper bound that defines our sampling probabilities. Intuitively, this allows us to obtain the subspace embedding guarantee for $\ell_2$ and $\ell_p$ simultaneously: the $\ell_2$ subspace embedding is known to hold already for small sample size $\tilde O(\eps^{-2}d)$. Taking it from here, it enables the $\ell_p$ subspace embedding to work with little more samples. Fortunately, this overhead is negligible compared to the small sample already taken, and also smaller compared to the previous bounds to obtain the $\ell_p$ subspace embedding directly.
Note that also reversely, the data matrix might have much smaller total $\ell_p$ sensitivity than $d$. In this case, augmenting the sample to at least the rank preserving lower bound can be accomplished with the least number of additional samples by $\ell_2$ sensitivity sampling.

The main reason why the previous analyses do not admit a simultaneous $\ell_2$ and $\ell_p$ subspace embedding, is that they tend to fold the weights into the data (or into the sampling matrix), rather than keeping the weights separately. This is a nice trick which simplifies wide parts of previous analyses by reducing the weighted case to the unweighted case. However, it prevents from our goals as we show in the following simple, yet instructive example for simultaneous $\ell_2$ and $\ell_1$ embedding, which requires to store the weights separately:

Take $A \in \mathbb{R}^{n \times 1}$ to be the matrix consisting of $n$ copies of the row vector $1$.
Note that for $x=1$ we have that $\norm{Ax}_1=n$ and $\norm{Ax}_2^2={n}$.
We wish to construct a subspace embedding, preserving both norms up to a factor of two.
To this end, assume that we have a reduced and reweighted matrix $A'\in\mathbb{R}^{m \times 1}$ with $\norm{A'x}_1\geq n/2 $.
Then we also have that $\norm{A'x}_2^2\geq \norm{A'x}^2_1/{m} \geq n^2/(4{m}) $.
Now, we also require that $\norm{A'x}_2^2 \leq 2{n} $. Combining both inequalities implies that $ 2{n} \geq n^2/(4{m})$ which is equivalent to $m \geq n/8 $.
We conclude that any subspace embedding without auxiliary weights that preserves both, the $\ell_1$, and the $\ell_2$ norm up to a factor of $2$ requires at least $\Omega(n)$ rows.

In stark contrast to this impossibility result, using probabilities $p_i \geq l_i^{(1)}+l_i^{(2)}$, standard sensitivity sampling allows us to take $\tilde O(\eps^{-2}d^2)$ samples $S$ and reweight them by $w$ such that $\sum_{i\in S} w_i|a_i x|=(1\pm\eps)\|A x\|_1$, and $\sum_{i\in S} w_i|a_i x|^2=(1\pm\eps)\|A x\|_2^2$ hold simultaneously.

To leverage this fact and improve the sampling complexity to linear, we need to open up and generalize large parts of the work of \citet{woodruffyasuda23} to deal with weighted $\ell_p$ norms, which we define as follows:

\begin{definition}
Given a vector $v \in \mathbb{R}^n$, $p \in [1, \infty)$ and weights $w \in \mathbb{R}_{\geq 1}^n$, we let
\[
    \norm{v}_{w, p}:= \left( \sum_{i=1}^n w_i |v_i|^p \right)^{1/p}.
\]
For the $\ell_{\infty}$ norm, we also let
\[
     \norm{v}_{w, \infty}:= \lim_{p \rightarrow \infty}\norm{v}_{w, p} = \max_{i \in [n]}|v_i|.
\]
and
\[
    \norm{v}_{w, \infty, p}:= \max_{i \in [n]}|w_i^{1/p} v_i|
\]
\end{definition}

\subsection{Outline of the analysis}
The full formal details can be found in the appendix. Here we provide an outline for the proofs of our main results. We note that some definitions or notation might slightly differ from the appendix for the sake of a clean and concise presentation. The proof consists of several main steps for which we give high level descriptions:

\subsubsection{Bounding by a Gaussian process} The first step is to bound the approximation error by a Gaussian process. Note that the subspace embedding guarantee of \cref{approx_guarantee} will be achieved when 
\begin{align*}
    \Lambda := \sup\nolimits_{x \in \mathbb{R}^d}|f(Ax)-f_w(SAx)| \leq \varepsilon f(Ax),
\end{align*}
where $f_w(Ax)=\sum\nolimits_{i=1}^n w_i h(a_i x)$, and $S\in\{0,1\}^{m\times n}$ is a sampling matrix with exactly one non-zero entry in each row such that $SA$ consists of $m$ rows out of $n$ from the data matrix $A$. By homogeneity of the loss functions, we can restrict the analysis to the case that $f(Ax)=1$, and our goal is to bound the above term by $\epsilon$.

To this end, we bound higher moments of the expected error using a standard symmetrization and Gaussianization argument by 
\[
\E_S |\Lambda|^l \leq 
    (2\pi)^{l/2}\E_{S, g} \sup\nolimits_{f(Ax)=1} \left|\sum_{i \in S}g_i w_i h(a_i x) \right|^l ,
\]
where the expectation is taken over the random subsample, represented by $S$, and on the right hand side additionally over i.i.d. standard Gaussians $g_i\sim N(0,1)$.

The sum on the right hand side is a Gaussian process that induces a pseudo metric, such that for all $y=SAx,y'=SAx'$, we have $d_X(y, y')\coloneqq$
\[
    \left(\E_{g \sim \mathcal{N}(0, I_{m})} \left(\sum_{i \in S}  g_i w_i h(y_i) - \sum_{i \in S}g_i w_i h( y_i'))  \right)^2\right)^{1/2}
\]
acting in the reduced and reweighted space.

For bounding the Gaussian process, we use a slightly adapted moment bound of \citet{woodruffyasuda23}
\[
     \E[\abs{\Lambda}^l] \leq (C\mathcal E)^l (\mathcal E/\mathcal D) + O(\sqrt l \mathcal D)^{l}.
\]
which follows from a tail bound due to Sudakov, which is sometimes attributed to Dudley. See \citep{Dudley2016,LT1991} for bibliographical discussion and references. 

Here, $C$ is a sufficiently large absolute constant, $\mathcal E$ is an upper bound on the entropy of the Gaussian process, and $\mathcal D$ is an upper bound on the diameter of the set $T_S\coloneqq T_S(A)=\{SAx \mid x\in \mathbb{R}^d, f(Ax)=1\}$ according to the pseudo-metric, i.e., $\sup\{d_X(y,y')\mid y,y'\in T_S(A)\}\leq \mathcal D.$

Note, that we can accomplish our goal by using the moment bound for an appropriately large choice of $l$ and applying Markov's inequality. Our remaining task thus reduces to bounding the diameter, and the entropy, and to quantify the required value of $l$, which will also determine a sufficient subsample size.

\subsubsection{Bounding the diameter}
We bound the diameter by relating it to the approximation error and to the largest possible coordinate in the reduced and reweighted $\ell_p$ norm vector among all vectors that satisfy $f(Ax)=1$. 
More specifically, let
\begin{align*}
    \sigma &:=\sup\nolimits_{f{(Ax)}=1} \norm{S A x}_{w, \infty, p}^p
\end{align*}
Note that $\sigma$ is similar to the largest $\ell_p$ leverage score.
Also define 
\[
    G := 1+\sup\nolimits_{f(Ax)=1}|f_{w}(SAx)-f(Ax)|.
\]
Then we prove that the diameter with respect to the pseudo metric is bounded by 
\[
    \diam(T_S(A)) \leq 4 (G \sigma)^{1/2}.
\]

We remark here, that this requires special care when treating the $p$-ReLU function $h(t)=\max\{0,t\}^p$ so not to loose additional $\mu$ factors unnecessarily. But in the final step, it can simply be upper bounded by the proper norm function since $\max\{0,t\}^p\leq |t|^p$ for all $t\in \mathbb{R}$. Thus, the same diameter bound applies to both functions without additional $\mu$ dependence.

\subsubsection{Bounds on covering numbers}
Before we can proceed with bounding the entropy of the Gaussian process, we first need to bound the smallest number of (weighted) $\ell_q$ balls of certain radius $t$ that are required to cover $\ell_p$ balls, for various values of $p,q\in [1,\infty)$.

To this end, we define balls according to the weighted norms.
Given a Matrix $A'\in \mathbb{R}^{m \times d}$, a weight vector $w \in \mathbb{R}^n$ and $q \geq  1$, we set 
\[
    B_w^q:=B_w^q(A')=\{ x \in \mathbb{R}^d ~|~ \norm{A'x}_{w, q} \leq 1 \}.
\]
For any $p,q\geq 1$ and $t > 0$, let the covering number $E(B_w^p,B_w^q,t)$ denotes the minimum cardinality of a set $N$ of $B_w^q$ balls of radius $t$ required to cover the unit $B_w^p$ ball. That is, $N$ is chosen such that for any $x \in B_w^p$ there exists $x' \in N$ with $\|x-x'\|_{w,q} \leq t$. This enables chaining arguments to construct a sequence of $t$-nets at different scales, which can be smaller than using one single $\eps$-net, i.e., one fixed scale for the entire space. See \citet{Nelson16} for a survey on chaining techniques and applications.

To bound the covering numbers, we aim at applying a so-called Dual-Sudakov-minoration result \citep[see][]{BLM1989}. Let $\norm*{\cdot}_X$ be a norm, and let $B\subseteq\mathbb R^d$ denote the Euclidean unit ball in $d$ dimensions. Then,
\[
    \log E(B, \norm*{\cdot}_X, t) \leq O(d)\frac{M_X^2}{t^2}
\]
where $M_X$ denotes the L{\'e}vy mean
\[
    M_X \coloneqq \frac{\E_{g\sim\mathcal N(0,I_d)}\norm*{g}_X}{\E_{g\sim\mathcal N(0,I_d)}\norm*{g}_2}.
\]
It is well known that the denominator is $\Theta(\sqrt d)$, therefore the previous bound reduces to 
\[
   \log E(B, \norm*{\cdot}_X, t) \leq O(1)\frac{(\E_{g\sim\mathcal N(0,I_d)}\norm*{g}_X)^2}{t^2} .
\]
A very important step in the proof is thus the following bound, which also required to be reproved to account for the weighted norm
for $q\geq 2$.

Consider an orthonormal matrix $A'\in\mathbb R^{m\times d}$ and $w \in \mathbb{R}_{\geq 1}^n$. Let $\tau \geq \max_{i=1}^m w_i\norm*{e_i^T A'}_2^2$ be the largest weighted $\ell_2$ leverage score. Then we have that 
\[
    \E_{g\sim\mathcal N(0,I_d)}\norm{A'g}_{w, q} \leq m^{1/q}\sqrt{q\cdot \tau}.
\]
Here is also where the $\ell_2$ leverage scores play a role in our bounds, and will later be used together with the $\ell_p$ leverage scores in order to balance the diameter bound with the entropy bound.

This bound allows us to control the covering numbers for various $\ell_q$ norms, including the $\ell_\infty$ norm for the respective weighted balls. In particular, we obtain bounds on the number of weighted $\ell_\infty$ balls required to cover weighted $\ell_p$ balls, by first covering $B_w^p$ using $B_w^2$ balls and then covering each $B_w^2$ ball again with $B_w^{\infty}$ balls. Applying a chaining technique using a telescoping sum over varying scales, yields a bound of roughly
\begin{align}\label{informal_covering_bound}
    \log E(B_w^p,d_X,t) = O(1)\frac{\tau \log m}{t^p}.
\end{align}

\subsubsection{Bounding the entropy}
For bounding the entropy, we need to control the following quantity
\[
    \int_0^\infty \sqrt{\log E(T_S, d_X, t)}~\mathrm{d}t.
\]
To this end, we first derive the following final covering bounds: one for small $t$ with a logarithmic dependence on $t^{-1}$ and a different bound for larger $t$ with a squared dependence on $t^{-1}$ but lower dependence on $d$.
\begin{align*}
    &1)~ \log E(T_S, d_{X}, t) \leq O(d) \log\left( \frac{G m}{t} \right),\\
    &2)~ \log E(T_S, d_{X}, t) \leq O(\log m) \frac{G^2 \tau}{t^2}.
\end{align*}
The first item follows by relating to the unweighted case, where a simple net construction suffices, i.e., 
\begin{align*}
    \log E(B_w^p, B_w^\infty, t) &\leq \log E(B^p, B^\infty, t) \leq  O(d) \log\left( \frac{G m}{t}\right).
\end{align*}
In fact, this is the only place in our proof where the weighted case can simply be reduced to the unweighted case. 
The second bound follows by applying the previous \cref{informal_covering_bound}.

We split the entropy integral at an appropriate point $\lambda$ into 
\[
    \int_0^\lambda \sqrt{\log E(T_S, d_X, t)}~\mathrm{d}t+\int_\lambda^{\mathcal D} \sqrt{\log E(T_S, d_X, t)}~\mathrm{d}t
\]
where the latter can be cut off at our previous diameter bound $\mathcal D$ because when the integral exceeds the diameter, it becomes $0$.

The two parts of the integral can now be bounded using the covering bounds $1)$ respectively $2)$ from above. That is, for small radii less than $\lambda$ we use the first bound and for radii larger than $\lambda$, we use the the second bound.

Choosing the right value for $\lambda$ so as to keep both terms appropriately small, we obtain the following entropy bound:
\[
    \int_0^\infty \sqrt{\log E(T_S, d_{X}, t)}\mathrm{d}t \leq O(G \tau^{1/2}){(\log m)}^{1/2}\log\frac{d\sigma}{\tau}
\]
Note, in particular the dependence on the weighted largest $\ell_2$ leverage score $\tau$ which will be crucial to balance the diameter with the entropy bound in the main proof.

\subsubsection{Outline of the main proof}
We have now worked out all pieces that we need in order to prove our main result given in \cref{thm:samplingthm_informal}. Again, we refer to the appendix for the full technical details. Here, we present a sketch of the final proof:

Let us begin with the sample size $m$. This is handled in a standard way by defining an indicator random variable that attains $X_i=1$ if row $i$ is in the sample and otherwise it attains $X_i = 0$. The expected size equals 
\[
    \E\left(\sum_{i=1}^n X_i \right)= \sum_{i=1}^n p_i = \alpha\left(1+\mu \sum_{i=1}^n l_i^{(p)} + \sum_{i=1}^n l_i^{(2)}\right).
\]
Since $d\geq 1$, $\sum_{i=1}^n l_i^{(p)}\leq d$, and $\sum_{i=1}^n l_i^{(2)}=d$, we can thus bound the expected size by 
\[
    \alpha d \leq \E\left(\sum_{i=1}^n X_i \right) \leq 3\alpha \mu d.
\]
An application of Chernoff's bounds yields in particular that $m\leq 6\alpha d \mu$ holds with probability at least $1-\delta$.

Next, note that by our choice of $\alpha$ it holds that $m\geq \tilde O(d+\log(1/\delta))$, and each sample is taken with probability $p_i \geq {l_i^{(2)}}$. This is sufficient to achieve the $ \ell_2$ norm subspace embedding up to a factor $1/2$ with probability at least $1-\delta$ \citep{Mahoney11}.

It allows us to relate $\tau$ to the largest weighted $\ell_2$ leverage score of the original matrix, rather than the subsample, i.e.,
\[
    \tau \leq 4\max_{i\in [n]} w_i l_i^{(2)}.
\]
We assume without loss of generality that $0 < p_i < 1 $ and thus noting that $\alpha l_i^{(2)} < p_i < 1$, we have that $ w_i l_i^{(2)}= l_i^{(2)}/p_i \leq 1/\alpha$ and thus $\tau \leq 4/\alpha$. A very similar argument yields $\sigma \leq 1/\alpha$. Consequently, we have that $$\sigma\approx\tau\approx 1/\alpha.$$

Plugging this into our diameter bound, we obtain
\[
    4(G  \sigma)^{1/2}
    \leq 8 (G/\alpha)^{1/2}
    \leq G \frac{\varepsilon}{2 \sqrt{l}} := \mathcal D,
\]
where $l=\Theta(\eps^2 \alpha)$ for a suitable constant.

Plugging this into the entropy integral bound, we obtain
\[
    \int_0^\infty \sqrt{\log E(T_S, d_X, t)}~\mathrm{d}t \leq G \varepsilon/8 := \mathcal E
\]
Finally, we found bounds for $\mathcal D$, and $\mathcal E$ which are suitably balanced and allow us to apply the moment bound of \citet{woodruffyasuda23}, which yields 
\[
     \E[\abs{\Lambda}^l] \leq (C'\mathcal E)^l (\mathcal E/\mathcal D) + O(\sqrt l \mathcal D)^{l}\leq C^l\eps^l\delta,
\]
for suitably large absolute constants $C'\leq C$.

Using this higher moment bound in an application of Markov's inequality, we get that $|\Lambda| \leq C \varepsilon$ holds with probability at least $1-\delta$, since
\[
    \Pr(|\Lambda| \geq C \varepsilon) = \Pr(|\Lambda|^l \geq C^l \varepsilon^l) \leq \frac{C^l \varepsilon^l \delta}{C^l \varepsilon^l} = \delta.
\]
This concludes the proof by taking a union bound over the three probabilistic events, and rescaling $\eps$ and $\delta$ respectively.

\section{Application to logistic regression}
Here we provide an outline and some high level intuition behind the proof of our second result given in \cref{thm:logistic_informal}.

The logistic loss function is given by
\[
    f(Ax)\coloneqq \sum_{i=1}^n \ln(1+\exp(a_ix)),
\]
so in our previous notation we have to deal with individual loss functions $h(r)=\ln(1+\exp(r))$.

Unfortunately, $f$ does not fully satisfy the assumptions of our main theorem. Therefore, we cannot apply \cref{thm:samplingthm_informal} directly.
Instead, we observe that $f$ can be rewritten in terms of the coordinate-wise ReLU function and the remainder.

We thus split $f$ into two parts $f(Ax)=f_1(Ax)+f_2(Ax)$ as follows:
\begin{align*}
    f(Ax)&=\sum_{i=1}^n \ln(1+\exp(a_ix)) \\
    &=\sum_{i=1}^n \ln(1+\exp(-|a_ix|)) + \sum_{i=1}^n  \max\{0, a_i x\} \\
    &=\sum_{i=1}^n h_1(a_ix) + \sum_{i=1}^n  h_2(a_i x)\\
    &=f_1(Ax) + f_2(Ax).
\end{align*}
Using this split, we show that taking a sample $S$ with probabilities $p_i=\min\{1, \alpha( \mu l_i^{(1)} + l_i^{(2)} + \mu d/n))\}$ where $\alpha= O(\frac{\log^3(\mu d \log(\delta^{-1})/\varepsilon)+\ln(\delta^{-1})}{\varepsilon^2})$ preserves the logistic loss function for all $x\in\mathbb{R}^d$ up to a relative error of at most $\varepsilon$. Note, that we oversample only by a linear factor $\mu$.

Indeed, this allows the $p$-ReLU function $h_2(r)=\max\{0,r\}^p$, with $p=1$ to be handled by a direct application of our main result, which yields 
\[
    \forall x\in \mathbb{R}^d\colon \left|f_2(Ax) - f_{2,w}(SAx)\right| \leq \eps f_2(Ax).
\]
The remaining part $f_1(Ax)$ is a bounded function, and can be handled by the uniform part of our sample. We note that this can be proven by a simple additive concentration bound, and charging the additive error by the optimal cost.

We take a slight detour using the standard sensitivity framework, which allows us to draw from existing previous work and saves a lot of technicalities. To this end, we restrict the function $h(r)$ to the negative domain to obtain $\tilde{h}\colon\mathbb{R}_{\geq 0}\rightarrow \mathbb{R}_{\geq 0}$ with $\tilde{h}(r)=\ln(1+\exp(-r))$.

Now, observe that $f_1$ can be rewritten as 
\[
    f_1(Ax) = f_1^{(1)}(Ax) + f_1^{(2)}(Ax) - 2 \frac{n}{\mu},
\]
where
\[
    f_1^{(1)}(Ax) = \frac{n}{\mu}+ \sum_{a_ix \geq 0} \tilde{h}(a_ix)
\]
and
\[
    f_1^{(2)}(Ax) = \frac{n}{\mu}+ \sum_{a_ix < 0} \tilde{h}(-a_ix).
\]
Next, we observe that since $\tilde{h}(r)\leq \tilde{h}(0)=\ln(2) <1$, and $f_1^{(1)},f_1^{(2)}\geq \frac{n}{\mu}$, we have that each sensitivity for both of the two functions is bounded by $\varsigma_i\leq \frac{\mu}{n}$, and the total sensitivity is $O(\mu)$.

Using the strict monotonicity of both functions, we can relate the VC dimension of a set system associated with the two functions to the VC dimension of affine hyperplane classifiers. Using a thresholding and rounding trick, this yields a final VC dimension bound of $O(d\log(\mu/\eps))$. These arguments are standard from recent literature, see \citet{MunteanuSSW18,MunteanuOP22} for details.

Putting the VC dimension and sensitivity bounds into the standard sensitivity framework and defining the approximate functions
\[
    f_{1,w}^{(1)}(SAx) = \frac{n}{\mu}+ \sum_{a_ix \geq 0} w_i \tilde{h}(a_ix)
\]
and similarly
\[
    f_{1,w}^{(2)}(SAx) = \frac{n}{\mu}+ \sum_{a_ix < 0} w_i \tilde{h}(-a_ix).
\]
yields for both $i\in\{1,2\}$ separately that
\[
    \forall x\in \mathbb{R}^d\colon \left|f_1^{(i)}(Ax) - f_{1,w}^{(i)}(SAx)\right| \leq \eps f_1^{(i)}(Ax),
\]
each with probability at least $1-\delta$.

Now, by a union bound, and using the triangle inequality we can put all three approximations together, which yields
\begin{align*}
        \bigg|f(Ax)-f_w(SAx)\bigg|&=\left| f_1^{(1)}(Ax) + f_1^{(2)}(Ax)-\frac{2n}{\mu} + f_2(Ax)  - f_{1, w}^{(1)}(SAx) - f_{1,w}^{(2)}(SAx) + \frac{2n}{\mu}  - f_{2, w}(SAx)  \right|\\
        &\leq \left|f_1^{(1)}(Ax) - f_{1,w}^{(1)}(SAx)  \right| + \left|f_1^{(2)}(Ax) - f_{1,w}^{(2)}(SAx) \right| + \left|f_2(Ax) - f_{2, w}(SAx) \right| \\
        &\leq \varepsilon \left(f_1^{(1)}(Ax)+ f_1^{(2)}(Ax) + f_2(Ax)\right)\\
        &\leq \varepsilon f(Ax)+2\eps \frac{n}{\mu}\\
        &\leq 3\varepsilon f(Ax).
\end{align*}
with probability at least $1-3\delta$. We conclude the proof by rescaling $\eps$, and $\delta$.

As a final remark, previous work aimed at approximating the coordinate-wise ReLU function $f_2(Ax)$ by bounding the error additively to within an $\eps$ fraction of the $\ell_1$ norm and then using a rescaled $\eps'=\eps/\mu$ to relate the $\ell_1$ norm back to the ReLU function.

Since the dependence on $\eps$ is typically quadratic, this approach is prone to a $\mu^2$ factor in the final sample size.
Our direct approximation of the ReLU function handled via \cref{thm:samplingthm_informal} requires oversampling by only $\mu l_i^{(p)}$, which results in a linear $\mu$ dependence for logistic regression as well.

\section{Conclusion and open directions}
In this paper, we resolve the sample complexity of $\ell_p$ subspace embedding via $\ell_p$ sensitivity sampling for all $p\in[1,2]$. 

Specifically, our work establishes new $\tilde\Omega(d^{2-p/2})$ lower bounds against pure $\ell_p$ leverage score sampling, showing that upper bounds implied by previous work of \citet{ChenD21} are tight in the worst case up to polylogarithmic factors.

By generalizing the approach of \citet{woodruffyasuda23} to deal with weighted norms and augmenting $\ell_p$ sampling probabilities with $\ell_2$ leverage scores, our work strengthens previous upper bounds of \citet{ChenD21,woodruffyasuda23} for all $p\in[1,2]$ to linear $\tilde O(\eps^{-2}(\mathfrak S^{(p)}+d))$ sample complexity, matching known $\tilde\Omega(d/\eps^2)$ lower bounds by \citet{LiWW21} in the worst case.

In particular, this resolves an open question of \citet{woodruffyasuda23} in the affirmative for $p \in [1,2]$ and brings the conceptually and computationally simple sensitivity subsampling approach into the regime that was previously only known to be possible using Lewis weights \citep{CohenP15}, or other alternatives \cite{JambulapatiLLS23}.

As an application of our results, we also obtain the first fully linear $\tilde O(\eps^{-2}\mu d)$ bound for approximating logistic regression, obtained via a special treatment of the $p$-generalization of the ReLU function, and improving over a previous $\tilde O(\eps^{-2}\mu^2 d)$ bound \citep{MaiMR21} as well as over $\tilde O(\eps^{-2}\mu d^c)$ bounds with $c\gg 1$. Our $p$-generalization of the ReLU function suggests similar improvements for $p$-generalized probit regression \citep{MunteanuOP22}, which we leave as an open problem.

We note that in the case of $p>2$, our generalization does not yield an improvement over the previous bounds of \citet{woodruffyasuda23}. Therefore, obtaining $\tilde O(\eps^{-2}(\mathfrak S^{(p)}+d))$ bounds or any improvement towards that goal remains an important open problem for future research. Further, we hope that our methods will help to improve sensitivity sampling bounds for other, more general loss functions, and distance based loss functions beyond $\ell_p$ norms.

\section*{Acknowledgements}
The authors would like to thank the anonymous reviewers of ICML 2024 for very valuable comments and discussion. This work was supported by the German Research Foundation (DFG), grant MU 4662/2-1 (535889065), and by the Federal Ministry of Education and Research of Germany (BMBF) and the state of North Rhine-Westphalia (MKW.NRW) as part of the Lamarr-Institute for Machine Learning and Artificial Intelligence, Dortmund, Germany. Alexander Munteanu was additionally supported by the TU Dortmund - Center for Data Science and Simulation (DoDaS).

\bibliography{arxivmain}
\bibliographystyle{apalike}

\newpage
\appendix
\onecolumn
\section{Setting}

We are given some dataset consisting of $n$ points $a_i \in \mathbb{R}^d$ for $i \in [n]$ where $d \ll n$.
We set $A \in \mathbb{R}^{n \times d}$ to be the matrix with rows $a_i$.
Further we are given a possibly weighted function $f_w(v)=\sum_{i=1}^n w_ih(a_i x)$ for some function $h: \mathbb{R} \rightarrow \mathbb{R}_{\geq 0}$ that measures the contribution of each point. We drop the subscript if the weights are uniformly $w_i=1$ for all $i\in [n]$

Specifically, we consider in this paper $h(r)= |r|^p$ and $h(r)=\max\{0, r\}^p$ for $p \in [1, \infty)$. We will also extend  the latter case for $p=1$ to logistic regression $h(r)=\ln(1+\exp(r))$.

Our goal is to show that when sampling with probability proportional to a mixture of the $\ell_p$-leverage scores, the $\ell_2$-leverage scores and $1/n$, then with a sample of $m\ll n$ elements we can guarantee with failure probability at most $\delta$ that it holds that
\begin{align}\label{goal}
    \sup\nolimits_{x \in \mathbb{R}^d}|f(Ax)-f_w(SAx)| \leq \varepsilon f(Ax),
\end{align}

where $S\in\{0,1\}^{m\times n}$ is a sampling matrix with exactly one non-zero entry in each row such that $SA$ extracts $m$ rows out of $n$ from the data matrix $A$, see \cref{def: sampling}.

We first prove a lower bound against pure $\ell_p$ leverage score sampling. The remainder is dedicated to proving our main results, namely optimal upper bounds for $\ell_p$ sensitivity sampling with $\ell_2$ augmentation.

\section{Lower bound against pure \texorpdfstring{$\ell_p$}{Lp} leverage score sampling}

\begin{theorem}\label{thm: lowerbound}
    There exists a matrix $A\in\mathbb{R}^{m\times 2d}$, for sufficiently large $m\gg 2d$, such that if we sample each row $i \in [n]$ with probability $p_i:= \min \{1, k l_i^{(p)}\}$ for some $k \in \mathbb{N}$, then with high probability, the $\ell_p$ subspace embedding guarantee (see \cref{lp_approx_guarantee}) does not hold unless $k=\Omega(d^{2-p/2}/(\log d)^{p/2})$.
\end{theorem}

\begin{proof}
Theorem 1.6 of \citep{woodruffyasuda23} implies the existence of an $n \times d$ matrix $A_1$ of full-rank $d$, where the dimensions $n\geq d$ are sufficiently large and $A_1$ has total $\ell_p$ sensitivity $O((d\log d)^{p/2})$.

Let $A_2$ be an $n' \times d$ matrix, where $n'\geq n$ is divisible by $d$, that consists of the identity matrix stacked $n'/d$ times. Note, that it has total $\ell_p$ sensitivity $d$ by construction. We let $m=n+n'$.

We combine the two matrices to get $A=\begin{bmatrix} A_1 & 0 \\ 0 & A_2 \end{bmatrix}$. Note that by construction, $A\in\mathbb{R}^{m\times 2d}$ has full rank $2d$ and its total $\ell_p$ sensitivity is $O((d\log d)^{p/2})+d\geq d$.

Obtaining an $\ell_p$ subspace embedding for $A$ requires to preserve the rank, which requires at least $d$ rows from each matrix, $A_1$ and $A_2$, to be sampled.

Sampling via pure $\ell_p$ scores as defined in the theorem, the probability to hit a row from $A_1$ is bounded above by $\kappa (d\log d)^{p/2}/d$ for some absolute constant $\kappa$.

Thus, if we sample less than $\frac{d^{2-p/2}}{4\kappa (\log d)^{p/2} }$ rows of $A$, the expected number of rows that we collect from $A_1$ is bounded by $$\frac{d^{2-p/2}}{4\kappa (\log d)^{p/2} } \cdot \frac{\kappa (d\log d)^{p/2}}{d} = d/4.$$

By a standard application of Chernoff's bound, our sample will thus comprise less than $d/2 < d$ rows from $A_1$ with high probability, implying that the subsample is rank deficient.

This proves an $\Omega(d^{2-p/2}/(\log d)^{p/2})$ lower bound against $\ell_p$ subspace embeddings via pure $\ell_p$ leverage score sampling.
\end{proof}

\section{Preliminaries}

\subsection{Covering numbers}

Covering numbers are the minimum numbers to cover one shape or body using multiple (possibly overlapping) copies of another shape or body. A prominent example is the minimum number of Euclidean balls of radius $1/2$ to cover the unit radius Euclidean ball.

More generally, if the second body is an $\varepsilon$-ball with respect to a certain metric then the covering number corresponds to the minimum size of an $\varepsilon$-net.

Given two convex bodies $T, K \subseteq \mathbb{R}^d$, we define the covering number $E(T, K)$ by
\[
    E(T, K)=\inf\nolimits_{k \in \mathbb{N}}\{ \exists X \subseteq \mathbb{R}^d, T \subseteq \bigcup_{x \in X}x+K   \}
\]
where $ x+K = \{ y=x+z ~|~ z \in K \}$ denotes the Minkowski sum.

Further, given a metric $d_X$ and $t \geq 0$, we define the ball of radius $t$, denoted $B_X(t)$, by
\[
    B_X(t)=\{ x \in \mathbb{R}^d ~|~ d_X(0, x) \leq t \}.
\]

Further we define $E(T, d_X, t)$ to be the minimum cardinality of any $t$-net of $T$ with respect to $d_X$, formally
\[
    E(T, d_X, t)=\inf \{|N| ~|~ N \subset T, \text{ for any $x \in T$ there exists $x' \in N$ with $d_X(x,x')\leq t$} \}
\]
If for any $ x', x \in T$ it holds that $d_X(x, x')\leq d_{X'}(0, x-x')$ then it holds that $E(T, d_X, t)\leq E(T, B_{X'}(t))$.

Given $\varepsilon>0$ we say that $N \subset \mathbb{R}^d $ is an $\varepsilon$-net of $T$ with respect to $d_X$ if for any point $x \in T$ there exists a point $y \in N$ such that $d_X(x, y)\leq \varepsilon$.
In particular, it holds for any $\varepsilon$-net $N$ that $E(T, d_X, \varepsilon)\leq |N|$.

\paragraph{Dual Sudakov Minoration}

The following result will help us bounding covering numbers for the case where $T$ is the Euclidean ball.

\begin{definition}[L{\'e}vy mean]
Let $\norm*{\cdot}_X$ be a norm. Then, the \emph{L{\'e}vy mean of $\norm*{\cdot}_X$} is defined as
\[
    M_X \coloneqq \frac{\E_{g\sim\mathcal N(0,I_d)}\norm*{g}_X}{\E_{g\sim\mathcal N(0,I_d)}\norm*{g}_2}.
\]
\end{definition}

Bounds on the L{\'e}vy mean imply bounds for covering the Euclidean ball by $\norm*{\cdot}_X$-balls using the following result:

\begin{theorem}[Dual Sudakov minoration, Proposition 4.2 of \citealp{BLM1989}]
\label{thm:dual-sudakov}
Let $\norm*{\cdot}_X$ be a norm, and let $B\subseteq\mathbb R^d$ denote the Euclidean unit ball in $d$ dimensions. Then,
\[
    \log E(B, \norm*{\cdot}_X, t) \leq O(d)\frac{M_X^2}{t^2}
\]
\end{theorem}

\subsection{Dudleys theorem}

The following moment bound follows from the so-called Dudley's theorem and is one of the main tools we use in our analysis: 

\begin{lemma}[\citealp{woodruffyasuda23}, slightly modified]
\label{lem:moment-bound}
Let $(X_t)_{t\in T}$ be a Gaussian process with pseudo-metric $d_X(s,t)\coloneqq \norm*{X_s - X_t}_2 = \sqrt{\E (X_s - X_t)^2}$.
Further let $\mathcal \diam(T) \coloneqq \sup\braces*{d_X(s,t) : s, t\in T}\leq \mathcal D$ be a bound on the diameter of $d_X$ on $T$ and $\int_0^\infty \sqrt{\log E(T, d_X, u)}~\mathrm{d}u \leq \mathcal E$ be a bound on the entropy.
Then, there is a constant $C = O(1)$ such that for
\[
    \Lambda \coloneqq \E\sup\nolimits_{t\in T}X_t
\]
it holds that
\[
     \E[\abs{\Lambda}^l] \leq (2\mathcal E)^l (\mathcal E/\mathcal D) + O(\sqrt l \mathcal D)^{l}.
 \]
\end{lemma}

The only modification is that we consider a more general $\Lambda$ than \citet{woodruffyasuda23}. However, we stress that their proof remains unchanged.

The moment bound can be used for a suitably large choice of $l$ to obtain low failure probability bounds using Markov's inequality.
The moment parameter $l$ has two functions: first, it allows us to remove the dependence on the term $(\mathcal E/\mathcal D)$, more precisely to replace it by $2^l$ which holds whenever $l \geq \log_2(\mathcal E/\mathcal D)$. 
Second, by choosing $l \geq \log(\delta^{-1})$ we can bound the failure probability by $\delta$ while having only an additive $\log(\delta^{-1})$ in the sample size.
The task then reduces to obtaining best possible bounds for the diameter $\mathcal D$ and the entropy $\mathcal E$ to allow a small sample size.

\subsection{Definitions}

Let $p \in [1, \infty]$ and $h: \mathbb{R}\setminus\{0\} \rightarrow \mathbb{R}_{\geq 0}$ be either the function with $h(y)=|r|^p$ or $h(y)=\max\{r, 0\}^p$.
Further given a data set $A \in \mathbb{R}^{n \times d}$ consisting of rows $a_1, \dots, a_n \in \mathbb{R}^d$ and a weight vector $w \in \mathbb{R}_{\geq 1}^n$ and $x \in \mathbb{R}^d$ we set
\[ f_w(Ax)= \sum_{i=1}^n w_i h(a_ix).\]

\begin{definition}[$\mu$-complex,\citealp{MunteanuSSW18,MunteanuOP22}, slightly modified]\label{def: mu-complex}
    Given a data set $A \in \mathbb{R}^{n \times d}$ we define the parameter $\mu$ as 
    \[\mu(A)=\sup\nolimits_{\norm{Ax}_p=1}\frac{\norm{Ax}_p^p}{\norm{(Ax)^-}_p^p} = \sup\nolimits_{\norm{Ax}_p=1}\frac{\norm{Ax}_p^p}{\norm{(Ax)^+}_p^p},
    \]
    where $(Ax)^+$, and $(Ax)^-$ are the vectors comprising only the positive resp. negative entries of $Ax$ and all others set to $0$.
\end{definition}

\begin{definition}[$\ell_p$-leverage scores]\label{def: l_p leverage scores}
    Given a data set $A \in \mathbb{R}^{n \times d}$ we define the $i$-th $\ell_p$ leverage score of $A$ by
    \[
    l_i^{(p)}=\sup\nolimits_{x \in \mathbb{R}\setminus\{0\}} \frac{|a_i x|^p}{\norm{Ax}_p^p}.
    \]
\end{definition}

\begin{definition}[sampling]\label{def: sampling}
    For $i \in [n]$ let $p_i \in (0, 1]$.
    Then sampling with probability $p_i$ is the following concept: We sample point $a_i$ with probability $p_i$ and set its weight to $w_i=1/p_i$.
    This sampling process can also be described using a matrix $S \in \{0 , 1\}^{m \times n}$ where $m$ is the number of sampled elements and $S_{ij}=1$ if and only if $j=j(i) \in [n]$ is the $i$-th sample.
    The weight vector $w \in \mathbb{R}_{\geq 1}^n$ is then defined by $w_i=w_{j(i)}$. We slightly overload the notation and identify with $S$ either the sampling matrix or the set of sampled indices.
\end{definition}
Note that, different to most previous work, we do not put the weights into the matrix $S$.

The following definition extends norms to work with (auxiliary) weights.

\begin{definition}[weighted norms]\label{def: norm}
Given a vector $v \in \mathbb{R}^n$, $p \in [1, \infty)$ and weights $w \in \mathbb{R}_{\geq 1}^n$, we let
\[
    \norm{v}_{w, p}:= \left( \sum_{i=1}^n w_i |v_i|^p \right)^{1/p}.
\]
For the $\ell_{\infty}$ norm, we also let
\[
     \norm{v}_{w, \infty}:= \lim_{p \rightarrow \infty}\norm{v}_{w, p} = \max_{i \in [n]}|v_i| \qquad \text{ and }\qquad \norm{v}_{w, \infty, p}:= \max_{i \in [n]}|w_i^{1/p} v_i|.
\]
\end{definition}

\section{Outline of the analysis}

The proof consists of several main steps, where the goal is to apply \cref{lem:moment-bound} to bound the deviation of the weighted subsample from the original function:
\begin{itemize}
    \item[1)] First, we show that the deviation of the weighted subsample from the original function can be bounded by a Gaussian process.
    \item[2)] We give a bound on the diameter of the Gaussian process of step 1).
    \item[3)] To bound the entropy, we first generalize some of the theory presented in \citep{woodruffyasuda23} to cope with auxiliary weights.
    \item[4)] Using the results of 2) and 3), we are able to bound the entropy of the Gaussian process of step 1).
    \item[5)] We then put everything together and proceed by proving the main theorem.
\end{itemize}

\section{Bounding by a Gaussian process}
We will analyze the following term:
\begin{align} \label{eqn:mainterm} 
\E_S \sup\nolimits_{f{(Ax)}=1}|f_w(SAx)-f(Ax)|^\ell
\end{align}
for some integer $\ell\geq 1 $.
Since both functions $h(r)=|r|^p $ and $h(r)= \max\{0, r\}^p$ are absolutely homogeneous, and we are interested in a relative error approximation, it suffices to consider points $x \in \mathbb{R}^d$ with $ f{(Ax)}=1$.
Towards applying Lemma \ref{lem:moment-bound}, we first bound above term by a Gaussian process:

\begin{lemma}\label{lem: gausbound}
    For $\ell \geq 1$ it holds that
    \[
    \E_S \sup\nolimits_{f{(Ax)}=1}|f_w(SAx)-f(Ax)|^\ell \leq 
    (2\pi)^{\ell/2}\E_{S, g} \sup\nolimits_{f{(Ax)}=1}\left|\sum_{i \in S} g_i w_i h(a_i x)\right|^\ell 
    \]
    where $g_i$ are independent standard Gaussians.
\end{lemma}
\begin{proof}
    
We first note that $|c+\cdot|^\ell$ is a convex function for any $c \in \mathbb{R}$ and the $\sup$ over convex functions is again convex. Thus by applying Jensen's inequality twice, we have that

\begin{align*}
    \E_S \sup\nolimits_{f_w(Ax)=1}|f_w(SAx)-f(Ax)|^\ell
    &= \E_S \sup\nolimits_{f(Ax)=1}|f_w(SAx)-f(Ax)+0|^\ell \\
    & = \E_S \sup\nolimits_{f(Ax)=1}|f_w(SAx)-f(Ax)+\E_{S'}(f(Ax)-f_{w'}(S'Ax))|^\ell \\
    &\leq \E_S \sup\nolimits_{f(Ax)=1}\E_{S'} |f_w(SAx)-f(Ax)+(f(Ax)-f_{w'}(S'Ax))|^\ell \\
    &\leq \E_{S, S'} \sup\nolimits_{f(Ax)=1}|f_w(SAx)-f_{w'}(S'Ax)|^\ell,
\end{align*}
where $S',w'$ are a second realization of our sampling matrix, and the corresponding weights.

The last term is bounded by 
\begin{align*}
    \E_{S, S'} \sup\nolimits_{f(Ax)=1}|f_w(SAx)-f_{w'}(S'Ax)|^\ell
    &=\E_{S, S'} \sup\nolimits_{f(Ax)=1} \left| \sum_{i \in S}w_i h(a_i x) - \sum_{i \in S'}w_ih(a_i x) \right|^\ell \\
    &\leq \E_{S, S', \sigma}\sup\nolimits_{f(Ax)=1} \left| \sum_{i \in S \cup S'}\sigma_i w_i h(a_i x) \right|^\ell
\end{align*}
  where $\sigma_i \in \{-1, 1\}$ are uniform random signs indicating whether $i$ are sampled by $S$ or $S'$. Terms that are sampled by both cancel and can thus only decrease the expected value. We may thus assume that no index is sampled in both copies $S,S'$.
  Further note that two copies of the same process can at most increase the expected value by a factor of two. Thus, we get that
\[
 \E_S \sup\nolimits_{f(Ax)=1}|f_w(SAx)-f(Ax)|^\ell \leq \E_{S, \sigma} \sup\nolimits_{f{(Ax)}=1}\left|2\sum_{i \in S}\sigma_i w_i h(a_i x)\right|^\ell=2^\ell \E_{S, \sigma} \sup\nolimits_{f(Ax)=1}\left|\sum_{i \in S}\sigma_i w_i h(a_i x)\right|^\ell.
\]
Finally, note that
\begin{align*}
    2^\ell \E_{S, \sigma} \sup\nolimits_{f(Ax)=1} \left|\sum_{i \in S}\sigma_i w_i h(a_i x) \right|^\ell \leq (2\pi)^{\ell/2}\E_{S, g} \sup\nolimits_{f(Ax)=1} \left|\sum_{i \in S}g_i w_i h(a_i x) \right|^\ell ,
\end{align*}
where $g_i$ are independent standard Gaussians, using a comparison between Rademacher and Gaussian variables (see, e.g., Lemma 4.5, and Equation 4.8 of \citealp{LT1991}).
\end{proof}

Given a realization $S$ we set $$G_S:=1+\sup\nolimits_{f(Ax)=1}|f_{w}(SAx)-f(Ax)|$$.
Then, we have that $ \sup\nolimits_{f(Ax)=1}|f_{w}(SAx)-f(Ax)| = G_S-1$.
We thus get the following corollary
\begin{corollary}\label{cor: gausbound}
    For $\ell \geq 1$ it holds that
    \[
    \E_S (G_S-1)^\ell \leq 
    (2\pi)^{\ell/2}\E_{S, g} \sup\nolimits_{f(Ax)=1} \left|\sum_{i \in S}g_i w_i h(a_i x) \right|^\ell .
    \]
\end{corollary}
Our goal is now to bound the right hand side of \cref{cor: gausbound} using Lemma \ref{lem:moment-bound}.
To this end, we will dedicate the following sections to bounding the diameter and the entropy of the Gaussian process $(2\pi)^{\ell/2}\E_{S, g} \sup\nolimits_{f(Ax)=1}|\sum_{i \in S}g_i w_i h(a_i x)|^\ell $.

\section{Analyzing the Gaussian pseudo metric}

In the following sections we consider a fixed realization $S$.
We set \begin{align}\label{eq:Gbound}
        G:=G_S=1+\sup\nolimits_{f(Ax)=1}|f_{w}(SAx)-f(Ax)|\geq \sup\nolimits_{f(Ax)=1} f_w(SAx) 
    \end{align} 
and for $y, y' \in \mathbb{R}^m$ we define
\[
    d_X(y, y')=\left(\E_{g \sim \mathcal{N}(0, I_{m})} \left(\sum_{i \in S}  g_i w_i h(y_i) - \sum_{i \in S}g_i w_i h( y_i'))  \right)^2\right)^{1/2}.
\]
\[
    T_A= \{ x \in \mathbb{R}^d ~|~ f(Ax)=1 \} \subseteq \{ x \in \mathbb{R}^d ~|~ \norm{Ax}_p^p \leq \mu \}
\]
and $\sigma = \max_i w_i l_i^{(p)}$.
Further we set $T_S=SA(T_A)= \{y \in \mathbb{R}^m ~|~ y=SAx, x \in T_A \}$

\subsection{Bounding the diameter}

Before bounding the diameter we prove the following lemma:

\begin{lemma}\label{lem: bounddXhelp}
    If $p \in (0, 2]$ for any $r, r' \in \mathbb{R}$ holds that
    \[
        ||r|^{p/2}-|r'|^{p/2}|\leq |r-r'|^{p/2} .
    \]
    and
    \[
        |\max\{0,r\}^{p/2}-\max\{0,r'\}^{p/2}|\leq |\max\{0,r\}-\max\{0,r'\}|^{p/2} .
    \]
\end{lemma}

\begin{proof}
    Let $0\leq a \leq b $.
    For $p \in (0, 2]$ note that $p/2-1 \leq 0$, and in the following calculation $t\geq t-a$. We thus have that
    \[
        b^{p/2}-a^{p/2}
        =\int_a^b (p/2) t^{p/2-1} ~ dt
        \leq \int_a^b (p/2) (t-a)^{p/2-1} ~ dt
        =\int_0^{b-a} (p/2) t^{p/2-1} ~ dt
        =(b-a)^{p/2}.
    \]
    Assuming without loss of generality that $ |r| \geq |r'|$, we can apply this followed by a triangle inequality to get
    \[
        |r|^{p/2}-|r'|^{p/2}\leq (|r|-|r'|)^{p/2} \leq |r-r'|^{p/2}.
    \]
    Similarly, if $r, r' \geq 0$, we have that
    \begin{align*}
        | \max\{0, r\}^{p/2}-\max\{0, r'\}^{p/2}|= |r|^{p/2}-|r'|^{p/2} \leq |r-r'|^{p/2}
        =|\max\{0, r\}-\max\{0,r'\}|^{p/2}
    \end{align*}
     or, if $r<0$ or $r'< 0$ holds, then 
    \[
        | \max\{0, r\}^{p/2}-\max\{0, r'\}^{p/2}|= \max\{0, r, r' \}^{p/2} \leq  |\max\{0, r\}-\max\{0,r'\}|^{p/2}.
    \]
\end{proof}

The following lemma bounds $d_X$ by the weighted infinity norm for $h(r)=|r|^p$. It allows us to deduce a bound on the diameter and we will later need it to bound the entropy as well.

\begin{lemma}\label{lem: bounddX}
    Let $p \in [1, 2]$ and let $h(r)=|r|^p$.
    For any $y, y' \in \mathbb{R}^m$ it holds that
    \begin{align}
        d_X(y, y')\leq  (2(f_w(y)+ f_w(y'))\norm{y-y'}_{w, \infty, p}^p)^{1/2}.
    \end{align}
    Further we have that
    \begin{align}
        d_X(y, 0)\leq (f_w(y)\norm{y}_{w, \infty, p}^p)^{1/2}.
    \end{align}
\end{lemma}

\begin{proof}
First note that since $\E(g_ig_j)=\mathbf{1}_{i=j}$ we have that
\[ 
    d_X(y, y')^2=  \sum_{i \in S} w_i^2(h(y_i)-h(y_i'))^2.
\]
Thus we have that
\begin{align*}
    d_X(y, y')^2&=  \sum_{i=1}^m w_i^2(h(y_i)-h(y_i'))^2\\
    &=  \sum_{i=1}^m w_i^2(h(y_i)^{1/2}-h(y_i')^{1/2})^2(h(y_i)^{1/2}+h(y_i')^{1/2} )^2\\
    & \stackrel{\text{\cref{lem: bounddXhelp}}}{\leq} \sum_{i=1}^m w_i^2 (|y_i-y_i'|^{p/2})^2(h(y_i)^{1/2}+h(y_i')^{1/2} )^2\\
    & \leq \norm{y-y'}_{w,\infty, 1/p}^p \cdot \sum_{i=1}^m w_i 2(h(y_i)+h(y_i'))) \\
    & \leq 2 \norm{y-y'}_{w,\infty, 1/p}^p (f_w(y)+ f_w(y')).
\end{align*}
For the second part of the lemma note that
\begin{align*}
    d_X(y, 0)^2&=  \sum_{i=1}^m w_i^2(h(y_i)-h(0))^2\\
    &=  \sum_{i=1}^m w_i^2h(y_i)^2
    \leq \sum_{i=1}^m w_i |y_i|^p w_i h(y_j)
    \leq f_w(y)\norm{y}_{w, \infty, p}^p.
\end{align*}
\end{proof}

We will now move our attention to $h(r)=\max\{0,r\}^p$.

\begin{lemma}\label{lem: bounddXReLU}
    Let $p \in [1, 2]$ and let $h(r)=\max\{0, r\}^p$.
    For any $y, y' \in \mathbb{R}^m$, let $(y)^+={(\max\{0, y_i\})}_{i=1}^{m}$ denote the vector that contains only the non-negative entries of $y$ and all others are set to $0$. Then, it holds that
    \begin{align}
        d_X(y, y')\leq  (2(f_w(y)+ f_w(y'))\norm{(y)^+-(y')^+}_{w, \infty, p}^p)^{1/2}.
    \end{align}
    Further we have that
    \begin{align}
        d_X(y, 0)\leq (f_w(y)\norm{(y)^+}_{w, \infty, p}^p)^{1/2}.
    \end{align}
\end{lemma}

\begin{proof}
First note that since $\E(g_ig_j)=\mathbf{1}_{i=j}$ we have that
\[ 
    d_X(y, y')^2=  \sum_{i \in S} w_i^2(h(y_i)-h(y_i'))^2.
\]
Thus we have that
\begin{align*}
    d_X(y, y')^2&=  \sum_{i=1}^m w_i^2(h(y_i)-h(y_i'))^2\\
    &=  \sum_{i=1}^m w_i^2(h(y_i)^{1/2}-h(y_i')^{1/2})^2(h(y_i)^{1/2}+h(y_i'){1/2} )^2\\
    & \stackrel{\text{\cref{lem: bounddXhelp}}}{\leq} \sum_{i=1}^m w_i^2 (|\max\{y_i, 0\}-\max\{y_i', 0\}|^{p/2})^2(h(y_i)^{1/2}+h(y_i')^{1/2} )^2\\
    & \leq \norm{(y)^+-(y')^+}_{w,\infty, 1/p}^p \cdot \sum_{i=1}^m w_i 2(h(y_i)+h(y_i'))) \\
    & \leq 2 \norm{(y)^+-(y')^+}_{w,\infty, 1/p}^p (f_w(y)+ f_w(y')).
\end{align*}
For the second part of the lemma note that
\begin{align*}
    d_X(y, 0)^2&=  \sum_{i=1}^m w_i^2(h(y_i)-h(0))^2\\
    &=  \sum_{i=1}^m w_i^2h(y_i)^2
    \leq \sum_{i=1}^m w_i \max\{0, y_i\}^p w_i h(y_i)
    \leq f_w(y)\norm{(y)^+}_{w, \infty, p}^p.
\end{align*}
\end{proof}

We thus conclude the following bound on the diameter.
\begin{lemma}\label{lem:diambound}
    Let $p \in [1, 2]$ It holds that $\diam(T_S) \leq 4 (G \sigma)^{1/2}$ where $\sigma =  \sup\nolimits_{f{(Ax)}=1} \norm{SA x}_{w,\infty, p}^p$.
\end{lemma}

\begin{proof}
    Let $y, y' \in T_S$.
    First note that there exist $x, x' \in T_A$ with $y=SAx $ and $y'=SAx'$. We thus have by \cref{eq:Gbound} that $f_w(y) \leq G.$
    
    Next note that
     \[
        \norm{(y)^+}_{w, \infty, p}^p
        \leq \norm{y}_{w, \infty, p}^p \leq \sup\nolimits_{f{(Ax)}=1} \norm{SA x}_{w,\infty, p}^p = \sigma.
    \]
    Similarly, we have that $ f_w(y') \leq G $ and $\norm{(y')^+}_{w, \infty, p}^p \leq  \sigma $.
    By \cref{lem: bounddX} and using the triangle inequality we have that
    \[
        d_X(y, y')\leq  (2(f_w(y)+ f_w(y'))\norm{y-y'}_{w, \infty, p}^p)^{1/2}
        \leq (8G\sigma)^{1/2} \leq 4 ( G\sigma)^{1/2}.
    \]
\end{proof}

\subsection{Bounds on covering numbers}

Next, we aim to bound the entropy. To this end, we first need to bound the log of the covering numbers $\log E(T, d_{X}, t) $.
We will use two bounds, one for small $t$ with a small dependence on $t^{-1}$ and a different bound for larger $t$ with larger dependence on $t^{-1}$ but smaller dependence on $d$.
For bounding the covering numbers we will use an approach similar to \citet{woodruffyasuda23}.
Given a Matrix $A'\in \mathbb{R}^{m \times d}$, a weight vector $w \in \mathbb{R}^n$ and $q >1$, we set $B_w^q:=B_w^q(A')=\{ x \in \mathbb{R}^d ~|~ \norm{A'x}_{w, q} \leq 1 \}$.

We will start with the following simple lemma that helps to gain a better understanding of covering numbers:

\begin{lemma}\label{lem: coveringnumberproperties}
Let $X \subseteq \mathbb{R}^n $ with $0 \in X $ and $r \in \mathbb{R}_{\geq 0}$ and let $rX=\{ rx ~|~ x \in X \}$. Further let $ \theta \geq 1$ and $d_X$ be a metric.
Then it holds that
\[
    \log E(rX, d_X^\theta, t)=\log E(X, r^\theta d_X^\theta, t)=\log E(X, d_X, (t/r)^{1/\theta}).
\]
Further, let $d_X$ be a metric with $d_X'(y, 0)\geq  d_X(y, 0)$ for all $y \in \mathbb{R}^d$. Then it holds that
\[
\log E(X, d_X, t)\leq \log E(X, d_X', t).
\]
\end{lemma}

\begin{proof}
    The first equality follows by homogeneity. For the second, notice that
    \begin{align*}
        \{ x\in \mathbb{R}^n ~|~  r^\theta d_X(0, x)^\theta \leq t \}
        =\{ x\in \mathbb{R}^n ~|~  d_X(0, x) \leq (t/r)^{1/\theta} \}.
    \end{align*}
    For the inequality, observe that
    \begin{align*}
        \{ x\in \mathbb{R}^n ~|~  d_X(0, x) \leq t \}
        \supseteq \{ x\in \mathbb{R}^n ~|~  d_X'(0, x) \leq t \}.
    \end{align*}
    Thus, we need more $t$-balls with respect to norm $d_X'$ sets to cover $X$ than we need $t$-balls with respect to $d_X$ and consequently $\log E(X, d_X, t)\leq \log E(X, d_X', t)$.
\end{proof}

We will now bound the number of $q$-balls we need to cover the Euclidean ball for $q \geq 2$.

We note that since $\E_{g\sim\mathcal N(0,I_d)}\norm*{g}_2 = \Theta(\sqrt{d})$, \cref{thm:dual-sudakov} implies $$\log E(B, \norm*{\cdot}_X, t) \leq O(1)\frac{(\E_{g\sim\mathcal N(0,I_d)}\norm*{g}_X)^2}{t^2}.$$

We thus proceed with a bound on the enumerator:

\begin{lemma}\label{lem:levy}
Let $q\geq 2$ and let $A'\in\mathbb R^{m\times d}$ and $w \in \mathbb{R}_{\geq 1}^m$. Let $\tau \geq \max_{i=1}^m w_i\norm*{e_i^T A'}_2^2$. Then,
\[
    \E_{g\sim\mathcal N(0,I_d)}\norm{A'g}_{w, q} \leq m^{1/q}\sqrt{q\cdot \tau}.
\]
\end{lemma}

\begin{proof}
We have for each row $a_i', i\in[n]$ that $a_i' g$ is distributed as a Gaussian random variable with zero mean and standard deviation $\|a_i'\|_2$. By a known bound for their $q$-th absolute moment, and applying the known upper bound on Stirling's approximation $\Gamma(x+1)\leq \sqrt{e\pi x}\left(\frac{x}{e}\right)^x$, we obtain
\begin{align*}
    \E_{g\sim\mathcal N(0,I_d)} w_i  \abs{a_i' g}^q &=  w_i \cdot \frac{2^{q/2}\Gamma(\frac{q+1}{2})}{\sqrt{\pi}}\norm{a_i'}_2^q \\
    &\leq w_i \cdot \frac{2^{q/2}}{\sqrt{\pi}}\cdot \sqrt{e\pi\frac{q-1}{2}}\left(\frac{q-1}{2e}\right)^{\frac{q-1}{2}}\cdot\norm{a_i'}_2^q\\
    &\leq w_i \cdot {2^{q/2}} \cdot e^{- \frac{q-1}{2} + \frac{1}{2}} \left(\frac{q-1}{2}\right)^{\frac{q-1}{2}+\frac{1}{2}}\cdot\norm{a_i'}_2^q\\
    &\leq w_i \cdot e^{- \frac{q-2}{2}} \left({q-1}\right)^{\frac{q}{2}}\cdot\norm{a_i'}_2^q\\
    &\leq w_i \cdot q^{q/2} \cdot \norm{a_i'}_2^q\,.
\end{align*}
Then by Jensen's inequality and linearity of expectation, we get
\begin{align*}
    \E_{g\sim\mathcal N(0,I_d)}\norm{A'g}_{w, q} &\leq \parens*{\E_{g\sim\mathcal N(0,I_d)}\norm{A'g}_{w, q}^q}^{1/q} \\
    &= \left(\sum_{i=1}^m w_i q^{q/2} \cdot \norm{a_i'}_2^q \right)^{1/q}\\
    &\leq \parens*{m\cdot q^{q/2} \cdot \max\nolimits_{i=1}^m w_i \norm{a_i'}_2^q}^{1/q} \\
    &\leq m^{1/q}\cdot q^{1/2} \cdot \max\nolimits_{i=1}^m w_i^{1/q} \norm{a_i'}_2\\
    &= m^{1/q}\sqrt{q\cdot \tau}\,.
\end{align*}
    
\end{proof}

By combining the above calculation with Theorem \ref{thm:dual-sudakov}, we obtain the following bound.
\begin{corollary}\label{cor:2-q-cover}
Let $2\leq q < \infty$ and let $A'\in\mathbb R^{m\times d}$ be orthonormal with respect to the weighted $\ell_2$ norm. Let $\tau \geq \max_{i=1}^m w_i\norm*{e_i^T A'}_2^2$. Then,
\[
    \log E(B_{w}^2(A'), \norm{.}_{w,q}, t) \leq O(1) \frac{m^{2/q} q\cdot \tau}{t^2}
\]
\end{corollary}

\begin{proof}
Since $A'$ is orthonormal with respect to the weighted $\ell_2$ norm, $B_{w}^2(A') = B_{w}^2$ is isometric to the Euclidean ball in $d$ dimensions, and thus Theorem \ref{thm:dual-sudakov} applies.
\end{proof}

We also get a similar result for $q = \infty$. To this end it suffices to apply \cref{cor:2-q-cover} with $q = O(\log n)$. 

\begin{corollary}\label{cor:2-inf-cover}
Let $A'\in\mathbb R^{m\times d}$ be orthonormal with respect to the weighted $\ell_2$ norm. Let $\tau \geq \max_{i=1}^m w_i\norm*{e_i^T A'}_2^2$. Then,
\[
    \log E(B_{w}^2, \norm{.}_{w, \infty}, t) \leq O(1) \frac{(\log m)\cdot \tau}{t^2}.
\]
\end{corollary}
\begin{proof}
We take a $B_{w}^q$ cover, for $q=2 \log m$, of $B_{w}^2(A')$ whose size is bounded in \cref{cor:2-q-cover} by at most $$C \frac{m^{2/q} q\cdot \tau}{t^2} = C\frac{2e\log(m)\cdot \tau}{t^2} \leq O(1) \frac{\log(m) \cdot \tau}{t^2} $$ for an absolute constant $C$. Now we replace every $q$-ball with an $\infty$-ball with the same center and radius $t$.
Note that for any $y \in\mathbb R^m$ we have that $ \norm*{y}_\infty \leq \norm*{y}_q$. This implies that $B_{w}^q\subseteq B_{w}^\infty$ for any fixed radius $t$. By this subset relation, the set of $\infty$-balls is a $B_{w}^\infty$ cover of $B_{w}^2(A')$ of the same size.
\end{proof}

By interpolation, we can improve the bound in Corollary \ref{cor:2-q-cover}:
\begin{lemma}\label{lem:2-r-cover}
Let $2 < r < \infty$ and let $A'\in\mathbb R^{m\times d}$ be orthonormal with respect to the weighted $\ell_2$ norm, for weights $w \in \mathbb{R}_{\geq 1}^m$. Let $\tau \geq \max_{i=1}^n w_i\norm*{e_i^T A'}_2^2$. Let $ t^{-1} \leq \poly(m)$. Then,
\[
    \log E(B_{w}^2, B_{w}^r, t) \leq O(1) \frac{1}{(t/2)^{2r / (r-2)}}\cdot \parens*{\frac{r}{r-2}\log d + \log m}\tau
\]
\end{lemma}
\begin{proof}

Let $q > r$, and let $0<\theta<1$ satisfy
\[
    \frac1r = \frac{1-\theta}2 + \frac\theta{q}.
\]
We define a measure $\nu: \mathcal{P}([n]) \rightarrow \mathbb{R}_{\geq 0}$ by $\nu(T)=\sum_{i \in T} w_i$ for $T \subseteq [n]$. 
Then by H\"older's inequality, we have for any $y \in\mathbb R^m$ that
\begin{align*}
    \norm*{y}_{w, r} 
    &= \parens*{\sum_{i=1}^m w_i\abs{y(i)}^{r(1-\theta)} \abs{y(i)}^{r\theta}}^{1/r} \\
    &= \parens*{\int_{[n]}\abs{y(i)^{2r}}^{(1-\theta)/2} \abs{y(i)^{rq}}^{\theta/q}~d\nu(i)}^{1/r}  \\
    &\leq \parens*{\int_{[n]}\abs{y(i)}^2 ~d\nu(i)}^{(1-\theta)/2}\parens*{\int_{[n]} \abs{y(i)}^q~d\nu(i)}^{\theta/q} \\
    & = \parens*{\sum_{i=1}^m w_i\abs{y(i)}^2}^{(1-\theta)/2}\parens*{\sum_{i=1}^m w_i \abs{y(i)}^q}^{\theta/q} 
    = \norm*{y}_{w, 2}^{1-\theta}\norm*{y}_{w,q}^\theta
\end{align*}
For any $y \in B_{w}^2$ we have that $ \Vert y \Vert_{w, 2} \leq 1$ and thus
\[
    \norm*{y}_{w, r}\leq \norm*{y}_{w, 2}^{1-\theta}\norm*{y}_{w, q}^\theta \leq \norm*{y}_{w, q}^\theta
    =\norm*{y}_{w, q}^\theta
\]
so by Lemma \ref{lem: coveringnumberproperties} and Corollary \ref{cor:2-q-cover}
\[
    \log E(B_{w}^2, B_{w}^r, t) \leq \log E(B_{w}^2, \norm*{.}_{w,q}^\theta, t) \leq \log E(B_{w}^2, B_{w}^q, t^{1/\theta}) \leq O(1)\frac{m^{2/q} q\cdot \tau}{t^{2/\theta}}.
\]
Now, we have
\[
    \frac2\theta = 2\frac{\frac12 - \frac1q}{\frac12 - \frac1r} = \frac{q-2}{q}\frac{2r}{r-2}
\]
so by taking $q = \frac{2r}{r-2} \log m$, we have that $m^{2/q} = \exp(\frac{ 2\ln(m)}{q})= O(1)$ and 
\[
    t^{2/\theta}=t^{(\frac{2r}{r-2} \log m-2)/\log m }=t^{{2r}/{(r-2)}}/t^{2/\log(m)}
    = \Theta(1) (t/2)^{2r/(r-2)}
\]
since $t^{-2/\log(m)}= \poly(m)^{2/\log(m)}=\Theta(1)$, so we conclude our claim.
\end{proof}

Using Lemma \ref{lem:2-r-cover}, we obtain the following analogue of Corollary \ref{cor:2-q-cover} for $p<2$. 

\begin{lemma}\label{lem:p-inf-cover}
Let $1 \leq p < 2$ and let $A' \in\mathbb R^{m\times d}$ be orthonormal with respect to the weighted $\ell_2$ norm. Let $\tau \geq \max_{i=1}^m w_i\norm*{e_i^T A'}_2^2$ and $t \geq 1/\poly(d) $. Then,
\[
    \log E(B_{w}^p, B_{w}^\infty, t) \leq O(1) \frac1{t^p}\parens*{\frac{\log d}{2-p} + \log m}\tau.
\]
Further if $2-p \leq \frac{1}{\ln(d)}$ then we have that
\[
    \log E(B_{w}^p, \norm{.}_{w, \infty}, t) = O(1) \frac{(\log m)\cdot \tau}{t^p}
\]
\end{lemma}

\begin{proof}
In order to bound a covering of $B_{w}^p$ by $B_{w}^\infty$, we first cover $B_{w}^p$ by $B_{w}^2$, and then use Corollary  \ref{cor:2-inf-cover} to cover $B_{w}^2$ by $B_{w}^\infty$.

We will first bound $E(B_{w}^p, B_{w}^2, t)$ using Lemma \ref{lem:2-r-cover}. For each $k\geq 0$, let $\mathcal E_k\subseteq B_{w}^p$ be a maximal subset of $B_{w}^p$ of minimum size such that for every pair of distinct $y,y'\in \mathcal E_k$, $\norm*{y-y'}_{w, 2} > 8^k t$, and for $8^{k}t \geq 1$ we define $\mathcal E_k \coloneqq \{0\}$ . Note that
\[
    \abs{\mathcal E_k} = E(B_{w}^p, B_{w}^2, 8^k t).
\]
Since for any point in $B_{w}^p$ and thus in particular for any point in $ \mathcal E_{k+1}$ there exists a point in $ \mathcal E_k$ by averaging, for each $k$, there exists $y^{(k)}\in\mathcal E_k$ such that if
\[
    \mathcal F_k \coloneqq \braces*{y\in\mathcal E_k : \norm{y - y^{(k)}}_{w, 2} \leq 8^{k+1}t},
\]
then
\[
    \abs{\mathcal F_k} \geq \frac{\abs{\mathcal E_k}}{\abs{\mathcal E_{k+1}}} = \frac{E(B_{w}^p, B_{w}^2, 8^{k}t)}{E(B_{w}^p, B_{w}^2, 8^{k+1}t)}
\]
We now use this observation to construct an $\ell_{p'}$-packing of $B_{w}^2$, where $p'$ is the H\"older conjugate of $p$. Let 
\[
    \mathcal G_k \coloneqq \braces*{\frac{1}{8^{k+1}t}(y - y^{(k)}) : y\in\mathcal F_k}.
\]
Then, $\mathcal G_k \subseteq B_{w}^2$ and since for $ y, y' \in \mathcal F_k$ it holds that $ \norm{y- y^{(k)}}_{w, p} \leq \norm{y}_{w, p} +\norm{y^{(k)}}_{w, p} \leq 2$ we also have that $\mathcal G_k \subseteq B_{w}^p \cdot 2 / 8^{k+1}t$. Further since $F_k \subseteq \mathcal E_k$ it holds that $\norm{y - y'}_{w, 2} > 1/8$ for every distinct $y,y'\in \mathcal G_k$. Then by H\"older's inequality,
\[
    \frac1{8^2} \leq \norm*{y-y'}_{w, 2}^2 \leq \norm*{y-y'}_{w, p} \norm*{y-y'}_{w, p'}
    \leq (\norm*{y}_{w, p}+\norm*{y'}_{w, p}) \norm*{y-y'}_{w, p'} 
    \leq \frac{4}{8^{k+1}t}\norm*{y-y'}_{w, p'}
\]
so $\norm*{y-y'}_{w, p'} \geq 2\cdot 8^{k-2}t$ which implies that the sets $S_y=\{ x \in B_2 ~|~ \norm*{x-y}_{w, p'} \leq 8^{k-2}t \} $ and $S_{y'}=\{ x \in B_2 ~|~ \norm*{x-y'}_{w, p'} \leq 8^{k-2}t \} $ are disjoint for any different $y, y' \in \mathcal G_k $.
Thus any maximal subset $S \subseteq \mathbb{R}^d$ such that for each distinct $y,y'\in \mathcal E_k$, $\norm*{y-y'}_{w, p'} > 8^{k-2}t$ must have at least one point in $S_y$ for any $y \in \mathcal G_k$.
Consequently
\begin{equation}\label{eq:dual-cover-bound}
    \log E(B_{w}^2, B_{w}^{p'}, 8^{k-2}t) \geq \log \abs{\mathcal G_k} = \log \abs{\mathcal F_k} \geq \log E(B_{w}^p, B_{w}^2, 8^{k}t) - \log E(B_{w}^p, B_{w}^2, 8^{k+1}t).
\end{equation}
Using that
\[
    \frac{p'}{p'-2}=\frac{(1-1/p)^{-1}}{(1-1/p)^{-1}-2}
    =\frac{(1-1/p)^{-1}}{(1-1/p)^{-1}(1-2(1-1/p))}
    =\frac{1}{2/p-1}=\frac{p}{2-p}
\]
and summing over $k$ gives
\begin{align*}
    \log E(B_{w}^p, B_{w}^2, t) &= \sum_{k\geq 0}\log E(B_{w}^p, B_{w}^2, 8^{k}t) - \log E(B_{w}^p, B_{w}^2, 8^{k+1}t) \\
    &\leq \sum_{k\geq 0} \log E(B_{w}^2, B_{w}^{p'}, 8^{k-2}t) && \text{\eqref{eq:dual-cover-bound}} \\
    &\leq O(1) \frac{1}{(t/2)^{2p' / (p'-2)}}\cdot \parens*{\frac{p'}{p'-2}\log d + \log m}\tau && \text{\cref{lem:2-r-cover}} \\
    &= O(1) \frac{1}{(t/2)^{2p / (2-p)}}\cdot \parens*{\frac{p}{2-p}\log d + \log m}\tau
\end{align*}
where we take $p' / (p'-2) = 1$ for $p' = \infty$.
Combining this with \cref{lem: coveringnumberproperties} and Corollary \ref{cor:2-inf-cover}, we now bound
\begin{align*}
    \log E(B_{w}^p, B_{w}^\infty, t) &\leq \log E(B_{w}^p, B_{w}^2,t') + \log E(B_{w}^2(t'), B_{w}^\infty, t) \\
    &\leq \log E(B_{w}^p, B_{w}^2, t') + \log E(B_{w}^2, B_{w}^\infty, t/t') \\
    &\leq O(1) \frac{1}{(t')^{2p / (2-p)}}\cdot \parens*{\frac{p}{2-p}\log d + \log m}\tau + O(1) \frac{(\log m)\cdot \tau}{(t/t')^2} \\
\end{align*}
for any $t'^{-1}\leq \poly(m)$. We choose $t'$ satisfying
\[
    \frac1{(t')^{2p/(2-p)}} = \frac{(t')^2}{t^2},
\]
which gives
\[
    t'=t^{2/(2+2p/(2-p))}=t^{2(2-p)/4}=t^{1-p/2}\geq t^{1/2}
\]
which implies that $t'^{-1}\leq \poly(m) $.
Further we get that
\[
    (t')^{2p/(2-p)} = \parens*{t^2}^{\frac{2p/(2-p)}{2 + 2p/(2-p)}} = t^p
\]
so we obtain a bound of
\[
    O(1) \frac1{t^p}\parens*{\frac{1}{2-p}\log d + \log m}\tau.
\]

For the final part of the lemma, first note that $B_p \subseteq B_2$ and thus by Corollary \ref{cor:2-inf-cover} and using that $t^{-1}=O(d^c) $ we have that
\begin{align*}
    \log E(B_{w}^p, \norm{.}_{w, \infty}, t) &\leq \log E(B_{w}^2, \norm{.}_{w, \infty}, t)\\
    &\leq O(1) \frac{(\log m)\cdot \tau}{t^2}\\
    &= O(1) \frac{(\log m)\cdot \tau}{t^p \cdot t^{2-p}}\\
    &= O(\exp (\log (d^c )/(2-p)) \frac{(\log m)\cdot \tau}{t^p }\\
\end{align*}
    Thus if $1/(2-p) = \Omega(\ln(d))$ then we have that $\log E(B_{w}^p, \norm{.}_{w, \infty}, t) = O(1) \frac{(\log m)\cdot \tau}{t^p}$.
\end{proof}

To deal with very small $t$ we will need another lemma.
We set $B_1^q(A)=\{ x \in \mathbb{R}^d ~|~ \norm{Ax}_{q} \leq 1 \}$, i.e., the case where the weights are uniformly $1$.

\begin{lemma}\label{lem: qr-coveringnumbers}
For any $1\leq r \leq q$, any weight vector $w \in \mathbb{R}^n_{\geq 1}$ and any $t\in\mathbb{R}_{\geq 1}$ it holds that
\[
    E(B_1^r(A), B_1^q(A), t) \geq E(B_w^r(A), B_w^q(A), t).
\]
\end{lemma}

\begin{proof}
    Assume that for $N \subseteq \mathbb{R}^{n}$ it holds that for any point in $y \in  B_1^r(A)$ there exists a point $x \in N$ such that $\norm{x-y}_q \leq t $.
    Given $x \in \mathbb{R}^n$ we define $x^{(q)}\in \mathbb{R}^n$ by $x_i^{(q)}=\frac{x_i}{w_i^{1/q}} $ and we set $N_w=\{x^{(q)} ~|~ x \in N \}$.
    Now let $y' \in B_w^r(A)$.
    We define $y\in \mathbb{R}^n$ by $y_i={w_i^{1/q}}\cdot y_i' $.
    Recall that $w_i=1/p_i \geq 1$, since any probabilities satisfy $p_i\leq 1$.
    We thus have that 
    \[\norm{y}_r^r=\sum_{i=1}^n |y_i|^r=\sum_{i=1}^n w_i^{r/q} |y_i'|^r\leq \sum_{i=1}^n w_i |y_i'|^r \leq 1\]
    since $y' \in B_w^r(A)$.
    Thus $y \in  B_1^r(A)$ and there exists $x \in N$ such that $\norm{x-y}_q \leq t $.
    Notice that
    \[
        |y_i'-x_i^{(q)}|=\left|\frac{y_i}{w_i^{1/q}}-\frac{x_i}{w_i^{1/q}}\right|=\frac{|y_i-x_i|}{w_i^{1/q}}.
    \]
    Then we have that
    \begin{align*}
        \norm{y'-x^{(q)}}_{w, q}^q=\sum_{i=1}^n w_i |y_i'-x_i^{(q)}|^q
        = \sum_{i=1}^n |y_i-x_i|^q \leq t^q
    \end{align*}
    and thus it holds that $ \norm{y'-x^{(q)}}_{w, q} \leq t$ and $N_w$ is a suitable net proving that $$ E(B_1^r(A), B_1^q(A), t) \geq E(B_w^r(A), B_w^q(A), t).$$\qedhere
\end{proof}

\subsection{Bounding the entropy}

Recall the original setting where $A \in \mathbb{R}^{n \times d}$, $T_A= \{ x \in \mathbb{R}^d ~|~ f(Ax)=1 \}$ 
and $T_S=SA(T_A)$, and $d_X(y, y')^2=  \sum_{i \in S} w_i^2(h(y_i)-h(y_i'))^2$.
Further let $m$ be the number of rows of $S$ and thus also of $A'=SA$.

Using the results from previous section we can deduce following bounds for the covering numbers of $T_S$ with respect to $d_X$:
\begin{lemma}\label{lem:coveringnumber}
Let $1\leq p\leq 2$. Then for any $t\in (0, 1]$ it holds that
\begin{align*}
    1)~ \log E(T_S, d_{X}, t) \leq O(d) \log\left( \frac{G m}{t} \right),\\
    2)~ \log E(T_S, d_{X}, t) \leq O(\log(m)) \cdot \frac{\gamma G^2  
    \tau}{t^2}.
\end{align*}
where $\tau$ is the maximum weighted $\ell_2$-leverage score of $(SA, w)$ and $\gamma= \frac{1}{2-p} $ for $ 2-p \geq 1/\ln(d)$ and $\gamma=1$ for $ 2-p \leq 1/\ln(d)$.
\end{lemma}

\begin{proof}
By Lemma \ref{lem: bounddX} and Lemma \ref{lem: bounddXReLU} it holds for all $y, y' \in T_S$ that $d_X(y, y') \leq (2(f_w(y)+ f_w(y'))\norm{y-y'}_{w, \infty, p}^p)^{1/2}$ for $h(r)=|r|^p$ and $d_X(y, y') \leq (2(f_w(y)+ f_w(y'))\norm{(y)^+-(y')^+}_{w, \infty, p}^p)^{1/2}$ for $h(r)=\max\{r, 0\}^p$.
For any $y, y' \in B_{w}^p(A')$ we thus have that $d_X(y, y') \leq 2\norm{y-y'}_{w, \infty, p}^{p/2}$.
We set $S_{w, p}$ to be the matrix we get by replacing each $1$ entry at column $i$ with $w_i^{1/p}$.
For $h(r)=|r|^p$ it holds that $T_S = B_{w}^p(A', G) \subseteq B^p(S_{w, 2}A, G)$ as
\[
    B_{w}^p(A', G)=\left\{ y \in \range(A') ~|~ \sum_{i}w_i y_i^p \leq 1 \right\}
    \subseteq \left\{ y \in \range(A') ~|~ \sum_{i}w_i^{1/2} y_i^p \leq 1 \right\}
    = B^p(S_{w, 2}A, G)
\]
since $w_i \geq 1$.
For $h(r)=\max\{r, 0\}$ it holds that $\{(y)^+ ~|~ y \in T_S  \} \subseteq B_{w}^p(A', G) \subseteq B^p(S_{w, 2}A, G)$ and by Lemma \ref{lem: bounddXReLU} it suffices to restrict to $(y)^+$ in this case.
Thus, rather than just covering the $1$-ball of $A'$, we need to cover the $G$-ball which is the same as covering the $1$-ball with an adjusted $\frac{t}{G }$ instead of $t$.
Thus, we have by Lemma \ref{lem: coveringnumberproperties} that
\begin{align*}
    \log E(T_S, d_{X}, t) 
    &\leq \log E(B^p(S_{w, 2}A),  2 \norm{\cdot}_{w,\infty, p}^{p/2}, t/G) \\
    &= \log E(B^p(S_{w, 2}A),  2 \norm{S_{w, 2}\cdot}_{w,\infty}^{p/2}, t/G) \\
    &= \log E\left(B^p(S_{w, 2}A), B^\infty(S_{w, 2}A), \left(\frac{t}{2 G }\right)^{2/p}\right).\\
    &= \log E\left(B^p(S_{w, 2}A), B^\infty(S_{w, 2}A), \left(\frac{t}{2 G }\right)^{2/p}\right).
\end{align*}

Next, to prove the claimed inequalities, we combine these bounds with bounds for the covering number $E(B^p(S_{w, 2}A), B_w^\infty(S_{w, 2}A), t) $.

For the first part of the lemma, it holds that $\log E(B_1^p, B_1^\infty, t) \leq O(d \log \frac{m}{t})$.
To see this, take an orthonormal basis $U$ of $S_{w, 2}A'$.
Then $N=\{ \frac{t}{{m}^{1/p}} U (n_1, \dots n_d)^T ~|~ \forall i\in [d] \colon n_i\in \mathbb{N}, n_i \leq m \}$ is a $t$-net of $B_1^p $ with respect to the $\ell_\infty$ norm.
By Lemma \ref{lem: qr-coveringnumbers} it holds that $\log E(B_w^p, B_w^\infty, t) \leq \log E(B_1^p, B_1^\infty, t)$.
Consequently we have that $\log E(T_S, d_{X}, t) \leq  O(d) \log\left( \frac{G m}{t}\right)$.

The second bound follows immediately by combining above argumentation with \cref{lem:p-inf-cover} as $\tau(S_{w, 2}A) = \tau(SA, w)$.
\end{proof}

For bounding the entropy, we will slightly adapt the proof of \citet{woodruffyasuda23} and use their following lemma.

\begin{lemma}[\citealp{woodruffyasuda23}]\label{lem:calc}
Let $0 < \lambda \leq 1$. Then,
\[
    \int_0^\lambda \sqrt{\log\frac1t}~\mathrm{d}t = \lambda \sqrt{\log(1/\lambda)} + \frac{\sqrt\pi}{4}\erfc(\sqrt{\log(1/\lambda)}) \leq \lambda \parens*{\sqrt{\log(1/\lambda)} + \frac{\sqrt\pi}2}
\]
\end{lemma}

Finally we are ready to bound the entropy:

\begin{lemma}\label{lem:entropy-int-p<2}
Let $1 \leq p < 2$ and let $SA'\in\mathbb R^{m\times d}$ be orthonormal with respect to the weighted $\ell_2$ norm. Let $\tau \geq \max_{i=1}^n \norm*{e_i^T SA'}_2^2$ and let $\sigma = \sup\nolimits_{i \in S, y \in T_S}w_i|y_i|^p$. Then if $\tau \in \Omega(\poly(1/d))$,
\[
    \int_0^\infty \sqrt{\log E(T_S, d_{X}, t)}~\mathrm{d}t \leq O(\gamma^{1/2} G \tau^{1/2})\parens*{\log m}^{1/2}\log\frac{d\sigma}{\tau}
\]
where $\gamma=\frac{1}{2-p} $ for $ 2-p \geq 1/\ln(d)$ and $\gamma=1$ for $ 2-p \leq 1/\ln(d)$.
\end{lemma}

\begin{proof}
Note that it suffices to integrate the entropy integral to $\diam(T_S)$ rather than $\infty$, since $$ \log E(T_S, d_{X}, t)=0$$ for $t \geq \diam(T_S)$ and recall that the diameter is at most $
4( G\sigma)^{1/2}$ by \cref{lem:diambound}, and since $p\geq 1$.

For small radii less than $\lambda$ for a parameter $\lambda$ to be chosen, we use the first bound of Lemma \ref{lem:coveringnumber}, i.e.
\[
    \log E(T_S, d_{X}, t) \leq O(d) \log\left( \frac{ G m}{t} \right)
\]
so by using Lemma \ref{lem:calc} we get that
\begin{align*}
    \int_0^\lambda \sqrt{\log E(T_S, d_{X}, t)}~\mathrm{d}t 
    = \int_0^\lambda \sqrt{O(d) \log\left( \frac{ G m}{t} \right)}~\mathrm{d}t 
    &\leq \lambda \sqrt{O(d)\log ( G m)} + \sqrt{O(d)}\int_0^\lambda \sqrt{\log\frac{1}{t}}~\mathrm{d}t \\
    &\leq \lambda \sqrt{O(d)\log ( G m)} + \sqrt d\parens*{\lambda \sqrt{\log\frac1\lambda} + \frac{\sqrt\pi}{2}\lambda}\\
    &\leq O(\lambda)\sqrt{d\log\frac{ G m}{\lambda}}
\end{align*}
On the other hand, for radii larger than $\lambda$, we use the the second bound of Lemma \ref{lem:coveringnumber}, which gives
\[
    \log E(T_S, d_X, t) \leq  O(\log m) \cdot \frac{\gamma G^2  \tau}{t^2}
\]
so the entropy integral gives a bound of
\[
    O(1)\bracks*{\parens*{ \log m}\gamma G^2 \tau}^{1/2}\int_\lambda^{4 (G \sigma)^{1/2}} \frac1t~\mathrm{d}t = O(1)\bracks*{\parens*{ \log m}\gamma G^2\tau}^{1/2} \log\frac{4 (G\sigma)^{1/2}}{\lambda}.
\]
We choose $\lambda = \sqrt{G^2\tau / d}=\Omega(\poly(1/d))$, which yields the claimed conclusion.
\end{proof}

\section{Proof of the main theorem}

\begin{theorem}\label{thm:samplingthm}
    Let $3/10 > \varepsilon, \delta >0$, $p \in [1, 2]$ and let $f(Ax)=\sum_{i=1}^n h(a_i x)$ where $h(r)=|r|^p$ or $h(r)=\max\{r, 0\}^p$.
    If
    $\alpha= O(\frac{ \gamma \log(d \mu \log(\delta^{-1})/\varepsilon )\ln^{2}(d)+\ln(\delta^{-1})}{\varepsilon^2})$
    and for all $i \in [n]$ it holds that $p_i \geq \min\{ 1, \alpha( \mu l_i^{(p)}+l_i^{(2)}+ \frac{1}{n})\}$.
        
    Then with failure probability at most $\delta$ it holds that 
    \begin{align*}
        \forall x \in \mathbb{R}^d\colon |f_w(SAx)-f(Ax)| \leq \varepsilon f(Ax)
    \end{align*}
    and the number $m$ of samples is bounded by
    \[
        m = O(d \mu\alpha)=O\left(\frac{d\mu}{\varepsilon^2}\left( \gamma \log(d \mu \log(\delta^{-1})/\varepsilon )\ln^{2}(d)+\ln(\delta^{-1})\right) \right)
    \]
    where $S$ and $w$ are constructed as in Definition \ref{def: sampling} (this corresponds to sampling point $i$ with probability $p_i$ and setting $w_i=1/p_i$), $\mu=1$ if $h(r)=|r|^p$ and $\mu=\mu(A)$ if $h(r)=\max\{r, 0\}^p$ and $\gamma=\frac{1}{2-p} $ for $ 2-p \geq 1/\ln(d)$ and $\gamma=1$ for $ 2-p \leq 1/\ln(d)$.
\end{theorem}

\begin{proof}
    First, without loss of generality we assume that for any $i \in [n]$ we have $p_i <1$. 
    
    Second, note that since $p_i>0$ for all $i$ we have for any $x \in \mathbb{R}^d$ that $\E(f_w(SAx))=f(Ax)$.

    Third, without loss of generality we assume that $p_i = \min\{ 1, \alpha(\mu l_i^{(p)}+l_i^{(2)}+ \frac{1}{n}) \}$ since increasing $p_i$ can only reduce the failure probability for obtaining the same approximation bound.

    By definition we have that $\alpha = O(\frac{ l }{\eps^2})$ where $l = O( \gamma \log(d \mu \log(\delta^{-1})/\varepsilon )\ln^{2}(d)+\ln(\delta^{-1})) $. 
    Assume the constants are large enough.

    We want to bound the number of samples $m$.
    To this end, we define the random variable $X_i=1$ if $i \in S$ and $X_i=0$ otherwise.
    Note that $X_i$ is a Bernoulli random variable and 
    using that the sum of sensitivities is bounded by $d$ for all $p < 2$ and equal to $d$ for $p=2$ \citep[cf.][]{woodruffyasuda23}, we get that
    \[
        \E\left(\sum_{i=1}^n X_i \right)= \sum_{i=1}^n p_i = \alpha\left(1+\sum_{i=1}^n \mu l_i^{(p)} + \sum_{i=1}^n l_i^{(2)}\right)=\alpha\left(1+\sum_{i=1}^n \mu l_i^{(p)} + d \right) \leq 3\alpha \mu d,
    \]
    and similarly, it also holds that $\E\left(\sum\nolimits_{i=1}^n X_i \right) \geq \alpha d.$
    
    An application of Chernoff bounds yields
    \[
        m =\sum_{i=1}^n X_i \leq 2\cdot\E\left(\sum_{i=1}^n X_i\right) \leq 6\alpha \mu d
    \]
     with failure probability at most $2\exp\left(-\E(\sum\nolimits_{i=1}^n X_i)/3 \right)\leq 2\exp(-\alpha d/3) \leq \delta$.

    We proceed with proving the correctness of our claim: by Corollary \ref{cor: gausbound} we have that
    \begin{align*}
        \E_S \sup\nolimits_{f{(Ax)}=1}|f_w(SAx)-f(Ax)|^\ell
        := \E_S (G_S-1)^\ell
     \leq  (2\pi)^{\ell/2}\E_{S, g} \sup\nolimits_{f{(Ax)}=1}\left|\sum_{i \in S}g_i w_i h(a_i x)\right|^\ell.
    \end{align*}

    For fixed $S$ we set $\Lambda \coloneqq \sup\nolimits_{f{(Ax)} = 1}\left|\sum_{i\in S}  g_i w_i h(a_i x)\right| .$
    We bound this quantity by Lemma \ref{lem:moment-bound} to get
    \[ \E_{g\sim\mathcal N(0,I_m)}[\abs{\Lambda}^l] \leq (2\mathcal E)^l (\mathcal E/\mathcal D) + O(\sqrt l \mathcal D)^{l} \]
    In the following we use the results of the previous sections to bound the entropy $\mathcal E$ and the diameter $\mathcal D $. To this end, we determine the parameters $\tau= \sup\nolimits_{\norm{SAx}_{w, 2}=1, i\in S}w_i \abs{a_i x}^2$ and $\sigma= \sup\nolimits_{f(Ax)=1, i\in [n]}w_i |a_i x|^p$.

    Taking $m\geq \tilde O(d+\log(1/\delta))$ samples with probability $p_i \geq {l_i^{(2)}}$ preserves the $ \ell_2$ norm up to a factor $1/2$ with failure probability at most $\delta$ \citep{Mahoney11}.
    We thus have that 
    \begin{align*}
        \tau&=\sup\nolimits_{\norm{SAx}_{w, 2}=1, i\in S}w_i \abs{a_i x}^2 \\
        &= \sup\nolimits_{\norm{SAx}_{w, 2}=1, i\in S} \frac{w_i \abs{a_i x}^2}{\norm{SAx}_{w,2}^2}\\
        &\leq \sup\nolimits_{\norm{SAx}_{w, 2}=1, i\in S} \frac{4 w_i \abs{a_i x}^2}{\norm{Ax}_{2}^2}\\
        &\leq \sup\nolimits_{\norm{Ax}_{2}=1, i\in [n]} \frac{4 w_i \abs{a_i x}^2}{\norm{Ax}_{2}^2}
        =4\max_{i\in [n]} w_i l_i^{(2)}.
    \end{align*}
    Now since $\alpha l_i^{(2)} < p_i < 1$ we have that $ w_i l_i^{(2)}= l_i^{(2)}/p_i \leq 1/\alpha$ and thus $\tau \leq 4/\alpha $. Similarly, since $p_i < 1$ we also have that
    \[
        \sigma = \max_{i \in [n]}\sup\nolimits_{f(Ax)=1}w_i |a_i x|^p  \leq \max_{i \in [n]} \frac{\mu l_i^{(p)}}{p_i}
        \leq \max_{i \in [n]} \frac{\mu l_i^{(p)}}{\alpha \mu l_i^{(p)}}
        \leq 1/\alpha.
    \]

    Thus, by Lemma \ref{lem:entropy-int-p<2}, using that $\sigma/\tau \leq 1$
    and choosing the constants for $\alpha$ large enough, we have that the entropy is bounded by
    \[
        O(\gamma^{1/2}G \tau^{1/2}(\log m)^{1/2}\log(d) )
        \leq O(\gamma^{1/2}G (1/\alpha )^{1/2}(\log m)^{1/2}\log(d) )
        \leq G \varepsilon/8 := \mathcal E
    \]
      
    Thus, we get by Lemma \ref{lem:diambound} a bound on the diameter of 
    \[
        4(G  \sigma)^{1/2}
        \leq 8 (G/\alpha)^{1/2}
        \leq G \frac{\varepsilon}{2 \sqrt{l}} := \mathcal D.
    \]
    
    Consequently, we get that
    \begin{align*}
        \E_{g\sim\mathcal N(0,I_m)}[\abs{\Lambda}^l] & \leq  (2\mathcal E)^l (\mathcal E/\mathcal D) + O(\sqrt l \mathcal D)^{l}\\
        &\leq (G \varepsilon/4)^l (\sqrt{l}/4) + O(G\varepsilon/2)^{l}\\
        & \leq (G \varepsilon/4)^l 2^l + O(G\varepsilon/2)^{l}\\
        & \leq G^l \varepsilon^l \delta
    \end{align*}
    
    Recall that $F_S= \sup\nolimits_{f(Ax)=1}|f_w(SAx)-f(Ax)|$. 
    Thus, we get in expectation over the sample $S$ that
    \begin{align*}
         F_S^l = (G_S-1)^l  \leq  3^lG_S^l\varepsilon^l \delta
         =3^l(1+F_S)^l \varepsilon^l \delta
         \leq 3^l\varepsilon^l \delta + 3^lF_S^l \varepsilon^l
    \end{align*}
    Rearranging the terms we get that
    \[  
        \E_S F_S^l-3^l\eps^l F_S^l \varepsilon^l \leq 3^l \delta.
    \]
    Dividing both sides of the inequality by $(1-3^l\varepsilon^l)$ and using that $\eps<3/10$ yields
    \[
        \E_S F_S^l \leq \frac{3^l \varepsilon^l \delta}{(1-3^l\varepsilon^l)}
        \leq 30^l \varepsilon^l \delta.
    \]
    Using Markov's inequality we get that $F_S^l \leq 30^l \varepsilon^l$ with failure probability at most $\delta$.
    \[
        \Pr(F_S \geq 30 \varepsilon) = \Pr(F_S^l \geq 30^l \varepsilon^l) \leq \frac{30^l \varepsilon^l \delta}{30^l \varepsilon^l} = \delta.
    \]
    To finish the proof, we recall that we need three events to hold: preserving the $\ell_2$ norm, the number of samples is $m = O(\alpha d)$ and $F_S \leq 30 \varepsilon $.
    The total failure probability for these events to hold is at most $\delta+\delta+\delta=3\delta$ by applying the union bound.
    Rescaling $\delta$ and $\varepsilon$ completes the proof.
\end{proof}

\section{Application to logistic regression}

\subsection{Sensitivity framework}\label{app:sensitivity}
We use the standard sensitivity framework \citep{LangbergS10} to handle a uniform sample. This requires first some terminology for the VC-dimension.
\begin{definition}
	The range space for a set $\mathcal{F}$ is a pair $\mathfrak{R}=(\mathcal{F},\ranges)$ where $\ranges$ is a family of subsets of $\mathcal{F}$. The VC-dimension $\Delta(\mathfrak{R})$ of $\mathfrak{R}$ is the size $|G|$ of the largest subset $G\subseteq \mathcal{F}$ such that $G$ is shattered by $\ranges$, i.e., $\left| \{G\cap R\mid R\in \ranges \} \right| = 2^{|G|}.$
\end{definition}
\begin{definition}
	Let $\mathcal{F}$ be a finite set of functions mapping from $\mathbb{R}^d$ to $\mathbb{R}_{\geq 0}$. For every $x\in\mathbb{R}^d$ and $r\in \mathbb{R}_{\geq 0}$, let $\rng{\mathcal{F}}(x,r) = \{ f\in \mathcal{F}\mid f(x)\geq r\}$, and $\ranges(\mathcal{F})=\{\rng{\mathcal{F}}(x,r)\mid x\in\mathbb{R}^d, r\in \mathbb{R}_{\geq 0} \}$, and $\mathfrak{R}_{\mathcal{F}}=(\mathcal{F},\ranges(\mathcal{F}))$ be the range space induced by $\mathcal{F}$.
\end{definition}

\begin{proposition}{\citep{FeldmanSS20}}
	\label{thm:sensitivity}
	Consider a family of functions $\mathcal{F}=\{f_1,\ldots,f_n\}$ mapping from $\mathbb{R}^d$ to $[0,\infty)$ and a vector of weights $u\in\mathbb{R}_{> 0}^n$. Let $\varepsilon,\delta\in(0,1/2)$.
	Let $s_i\geq \zeta_i$.
	Let $S=\sum\nolimits_{i=1}^{n} s_i \geq \sum\nolimits_{i=1}^{n} \zeta_i = Z$. Given $s_i$ one can compute in time $O(|\mathcal{F}|)$ a set $R\subset \mathcal{F}$ of $$O\left( \frac{S}{\varepsilon^2}\left( \Delta \ln S + \ln \left(\frac{1}{\delta}\right) \right) \right)$$ weighted functions such that with probability $1-\delta$, we have for all $x\in \mathbb{R}^d$ simultaneously $$\left| \sum_{f_i\in \mathcal{F}} u_i f_i(x) - \sum_{f_i\in R} w_i f_i(x) \right| \leq \varepsilon \sum_{f_i\in \mathcal{F}} u_i f_i(x),$$
	where each element of $R$ is sampled i.i.d. with probability $p_j=\frac{s_j}{S}$ from $\mathcal{F}$, $w_i = \frac{Su_j}{s_j|R|}$ denotes the weight of a function $f_i\in R$ that corresponds to $f_j\in\mathcal{F}$, and where $\Delta$ is an upper bound on the VC-dimension of the range space $\mathfrak{R}_{\mathcal{F}^*}$ induced by $\mathcal{F}^*$ obtained by defining $\mathcal{F}^*$ to be the set of functions $f_j\in\mathcal{F}$, where each function is scaled by $\frac{Su_j}{s_j|R|}$.
\end{proposition}

\subsection{Sensitivity sampling for logistic regression}

The logistic loss function is given by
\[
    f(Ax)=\sum_{i=1}^n h(a_ix)
\]
where $h(r)=\ln(1+\exp(r))$.
Since $f$ does not satisfy our assumptions, we cannot apply our main theorem directly.
Instead we split $f$ into two parts:
\begin{align}\nonumber
    f(Ax)=\sum_{i=1}^n \ln(1+\exp(a_ix)) &= \sum_{a_ix<0} \ln(1+\exp(a_ix)) ~+~ \sum_{a_ix \geq 0} \ln(1+\exp(a_ix))\\
    \nonumber
    &=\sum_{a_ix<0} \ln(1+\exp(a_ix)) ~+~ \sum_{a_ix \geq 0} \ln(\exp(a_ix)(\exp(-a_ix)+1))\\
    \nonumber
    &=\sum_{a_ix<0} \ln(1+\exp(a_ix)) ~+~ \sum_{a_ix \geq 0} \ln(\exp(-a_ix)+1) + a_i x\\
    \label{eq:logregsplit}
    &=\sum_{i=1}^n \ln(1+\exp(-|a_ix|)) + \sum_{a_ix \geq 0}  a_i x.
\end{align}

Using this split we show that sampling with probabilities $p_i=\min\{1, \alpha( \mu l_i^{(1)} + l_i^{(2)} + \mu d/n))\}$ where $\alpha= O(\frac{\log^3(\mu d \log(\delta^{-1})/\varepsilon)+\ln(\delta^{-1})}{\varepsilon^2})$ preserves the logistic loss function for all $x\in\mathbb{R}^d$ up to a relative error of at most $\varepsilon$.

\begin{theorem}\label{thm:logistic}
     Let $A\in \mathbb{R}^{n\times d}$ be $\mu$-complex some $\infty > \mu \geq 1$ and let $\varepsilon, \delta \in (0, 1/2)$.
     Further assume that we sample with probabilities $p_i \geq \min\{1,\alpha (\mu l_i^{(1)}+l_i^{(2)}+\frac{\mu d}{n})\}$ for $\alpha \geq O( (\log^3(d\mu\log(\delta^{-1})/\varepsilon ) + \ln(\delta^{-1}) ) /\varepsilon^{2})$, where the number of samples is 
     \[
        m= O\left(\frac{d\mu}{\varepsilon^{2}} \left( \log^3(d \mu \log(\delta^{-1})/\varepsilon) + \log(\delta^{-1}) \right) \right)
     \]
     Then we have that
     \[
        \forall x \in \mathbb{R}^d\colon |f_w(SAx)-f(Ax)| \leq \varepsilon f(Ax)
     \]
     with failure probability at most $\delta$.
\end{theorem}

\begin{proof}
    Note that by applying our main result given in \cref{thm:samplingthm} for $h(t)=\max\{0,t\}$, we get with failure probability at most $\delta/3$ that 
    \begin{align*}
        \forall x\in\mathbb{R}^d\colon \left|\sum_{a_ix \geq 0, i \in [n]}a_ix - \sum_{a_ix \geq 0, i \in S}w_i a_ix  \right| \leq \varepsilon \sum_{a_ix \geq 0, i \in [n]}a_ix.
    \end{align*}
    This handles the second term in our split of \cref{eq:logregsplit} within our framework.

    Further consider for the first sum of \cref{eq:logregsplit}, the functions $h: \mathbb{R}_{\geq 0}\rightarrow \mathbb{R}_{\geq 0}$, $h(r)= \ln(1+\exp(-r))$, $f_1(Ax)= \frac{n}{\mu}+ \sum_{a_ix \geq 0} h(a_ix)$ and $f_2(Ax)= \frac{n}{\mu}+ \sum_{a_ix < 0} h(-a_ix)$.
    
    Our goal is to apply Proposition \ref{thm:sensitivity} to $f_1$ and $f_2$.
    
    We first note that $h(r) \leq h(0)=\ln(2) \leq 1$ and $f_1(Ax),f_2(Ax)  \geq \frac{n}{\mu}  $ and thus the sensitivities are bounded for both functions by $\zeta_i \leq \frac{\mu}{n}$. The total sensitivity is consequently bounded by $O(\mu)$.
    
    Next observe that the VC dimension of the range spaces of $\mathcal{F}_1=\{ h_z : x \mapsto h(zx) \}$ and $\mathcal{F}_1=\{ h_z : x \mapsto h(-zx) \}$ are bounded by $ d+1$ since $h $ is an increasing function which allows to relate to the VC dimension of affine hyperplane classifiers by a standard argument, \citep[cf.][]{MunteanuSSW18,MunteanuOW21}.
    
    Further, applying a thresholding and rounding trick to the sensitivities \citep{MunteanuOP22}, the VC dimension of the weighted range spaces are bounded by $O(d\log(\mu/\varepsilon))$.

    Since $p_i \geq \frac{\mu}{n}$ and $m \geq \Omega(\frac{\mu}{\eps^2} \cdot (d \ln(\mu/\varepsilon)\ln(\mu)+ \ln(1/\delta))$ we have by \cref{thm:sensitivity} with failure probability at most $\delta/3$
    \begin{align*}
        \left|f_{w,1}(SAx) - f_1(Ax) \right| \leq \varepsilon f_1(Ax).
    \end{align*}
    Similarly with failure probability at most $\delta/3$ it holds that
    \begin{align*}
        \left|f_{w,2}(SAx) - f_2(Ax) \right| \leq \varepsilon f_2(Ax).
    \end{align*}
    Now combining everything by triangle inequality, and using the union bound, we get with failure probability at most $\delta$ that
    \begin{align*}
        &|f(Ax)-f_w(SAx)|\\
        &=\left| \sum_{a_ix \geq 0, i \in [n]}a_ix +f_1(Ax)-\frac{n}{\mu} +f_2(Ax)-\frac{n}{\mu} - \sum_{a_ix \geq 0, i \in S}w_i a_ix - f_{w, 1}(SAx) + \frac{n}{\mu} -  f_{w, 2}(SAx) +\frac{n}{\mu} \right|\\
        &\leq \left|\sum_{a_ix \geq 0, i \in [n]}a_ix - \sum_{a_ix \geq 0, i \in S}w_i a_ix  \right| + \left|f_1(Ax) - f_{w, 1}(SAx) \right|+ \left|f_2(Ax) - f_{w, 2}(SAx) \right| \\
        &\leq \varepsilon \left(\sum_{a_ix \geq 0, i \in [n]}a_ix+ f_1(Ax) + f_2(Ax )\right)\\
        &\leq \varepsilon f(Ax)+2\eps \frac{n}{\mu} \leq 3\varepsilon f(Ax).
    \end{align*}
    The last inequality follows from the lower bound $f(Ax)\geq \frac{n}{\mu}$ of \citet{MunteanuOW21}. Rescaling $\eps$ concludes the proof.
\end{proof}

\end{document}